\documentclass[preprint,12pt,3p]{elsarticle}

\usepackage{soulutf8}

\soulregister\cite7
\soulregister\ref7

\usepackage{amssymb}
\usepackage{amsmath,amsthm}
\usepackage{textcomp}
\usepackage{array}
\usepackage{color}
\usepackage{setspace}
\usepackage[]{datetime}
\usepackage{hyperref}
\usepackage{tikz}
\usepackage{subfig}
\usepackage{graphicx}
\usepackage{pgfplots}
\pgfplotsset{compat=1.13}
\usepackage{epsfig}
\usepackage{tabu}

\newcommand{\trio}{\mathrm{T{\footnotesize ROMINO}}}
\newcommand{\strio}{\mathrm{180\mbox{-}T{\footnotesize ROMINO}}}
\newcommand{\integer}{\mathbb{Z}}
\newcommand{\ar}{\mathcal{AR}}
\newcommand{\arr}{\mathcal{AR}_{a,b}}
\newcommand{\nat}{\mathbb{N}}
\newcommand{\ad}{\operatorname{AD}}
\newcommand{\hs}[1]{\hspace*{ #1 mm}}

\newtheorem{theorem}{Theorem}
\newtheorem{lemma}[theorem]{Lemma}
\newtheorem{corollary}[theorem]{Corollary}
\newdefinition{definition}[theorem]{Definition}
\newdefinition{remark}[theorem]{Remark}

\definecolor{aLittleDarkerGreen}{RGB}{0, 150, 50}

\usepackage[normalem]{ulem} 

\usepackage{fancyhdr}
 
\journal{Theoretical Computer Science}

\begin{document}

\begin{frontmatter}

\title{Hard and Easy Instances of L-Tromino Tilings\tnoteref{t1}}

\tnotetext[t1]{An extended abstract of this paper appeared in the Proceedings of WALCOM 2019 \cite{WALCOM}.}


\author[label1]{Javier T. Akagi}
\ead{akagi.tada@gmail.com}

\address[label1]{NIDTEC, Universidad Nacional de Asunci\'on, Campus Universitario, San Lorenzo C.P. 2619, Paraguay}

\author[label1]{Carlos F. Gaona}

\author[label1]{Fabricio Mendoza}

\author[label2]{Manjil P. Saikia\fnref{label3,fn3}}
\address[label2]{Fakult\"at f\"ur Mathematik, Universit\"at Wien,  Oskar-Morgenstern-Platz 1, 1090 Vienna, Austria}
\ead{manjil@gonitsora.com, manjil.saikia@univie.ac.at}
\ead[url]{https://manjilsaikia.in/}
\fntext[label4]{Supported by the Austrian Science Foundation FWF, START grant Y463 and SFB grant F50.}

\fntext[fn3]{Current Address: School of Mathematics, Cardiff University, Cardiff, CF24 4AG, UK.}

\author[label1]{Marcos Villagra\corref{cor1}\fnref{label5}}
\ead{mvillagra@pol.una.py}
\cortext[cor1]{Corresponding author}
\ead[url]{http://www.cc.pol.una.py/~mvillagra}
\fntext[label5]{Supported by Conacyt research grant PINV15-208.}

\begin{abstract}
We study tilings of regions in the square lattice with L-shaped trominoes. Deciding the existence of a tiling with L-trominoes for an arbitrary region in general is NP-complete, nonetheless, we identify restrictions to the problem where it either remains NP-complete or has a polynomial time algorithm. First, we characterize the possibility of when an Aztec rectangle and an Aztec diamond has an L-tromino tiling. Then, we study tilings of arbitrary regions where only $180^\circ$ rotations of L-trominoes are available. For this particular case we show that deciding the existence of a tiling remains NP-complete; yet, if a region does not contains certain so-called ``forbidden polyominoes'' as sub-regions, then there exists a polynomial time algorithm for deciding a tiling.
\end{abstract}


\end{frontmatter}

\section{Introduction}
\subsection{Background}
A packing puzzle is a solitary game where a player tries to find a way to cover a given shape using polyominoes, where a polyomino is a set of squares joined together by their edges. The computational complexity of packing puzzles was studied by Demaine and Demaine \cite{DD07} and they showed that tiling a shape or region using polyominoes is NP-complete.

In this paper we study tilings of regions in the square lattice with L-shaped trominoes (a polyomino of three cells) called an \emph{L-Tromino} or simply \emph{tromino} in this work. A cell in $\integer^2$ is a subset $[a,a+1]\times [b,b+1]$ and a region is any finite union of connected cells. At our disposal we have an infinite number of trominoes and would like to know if a given region can be covered or tiled with trominoes. 

The problem of tiling with trominoes was first studied by Conway and Lagarias \cite{CL90} who presented an algebraic necessary condition for a region in order to have a tiling. Moore and Robson \cite{MR01} showed that deciding if a region can be covered with trominoes is NP-complete. Later Horiyama \emph{et al.} \cite{HIN17} presented another proof of NP-completeness by constructing an one-one reduction which implies that counting the number of tilings with trominoes is \#P-complete. Counting the number of tilings with L-trominoes was also studied by Chin \emph{et al.} \cite{CGH07} using generating functions.

\subsection{Contributions}
In this paper we aim at identifying instances of the tiling problem with trominoes that either have efficient algorithms
or remain NP-complete. As a further generalization of the problem, we also consider regions with ``defects'' or holes, 
that is, we want to know if there is a tiling with trominoes without covering the defects. First we study the 
Aztec rectangle (and hence, also an Aztec diamond) \cite{EKLP92,MPS17} and show that any Aztec rectangle of side lengths $a,b$ can be covered with trominoes if and only if $a(b+1)+b(a+1)\equiv 0 \pmod 3$ (Theorem \ref{th-1}), which implies the existence of a polynomial time algorithm for 
finding a tiling in an Aztec rectangle. Then we show that for the cases when
$a(b+1)+b(a+1)\equiv 0 \pmod 3$ does not hold, if an Aztec Rectangle has exactly one defect, then it can be covered with trominoes
(Theorem \ref{the:ar-defect}). In general, however, deciding the tiling of an Aztec diamond with an unknown number of defects is 
NP-complete (Theorem \ref{the:az-defect-hard}).

In the second part of this paper we study a restricted case of the tiling problem where we only have $180^\circ$ rotations 
of the trominoes available. Here we show that the problem remains NP-complete (Theorem \ref{the:180tromino}) by slightly 
modifying the one-one reduction from the 1-in-3 Graph Orientation Problem of Horiyama \emph{et al.} \cite{HIN17}, whereas 
any Aztec rectangle has no tiling at all (Theorem \ref{the:Aztec-180}). Nevertheless, we show that if a region does not contain
any of the so-called ``forbidden polyominoes'' identified in this work, then that region has an efficient algorithm 
for deciding a tiling (Theorem \ref{the:forbidden}). This latter result is proved by constructing a planar dual graph 
of the region, called an intersection graph, and identifying independent sets of certain size. If the intersection graph 
has a claw, then that claw will correspond to a forbidden polyomino; if the graph is claw-free, however, we can use 
well-known efficient algorithms for finding independent sets, and hence, a tiling for the region.

In Section \ref{sec:li} we study a relation between L-Trominoes and I-Trominoes (a tromino with the shape of an I). We introduce a technique for decomposing a region in simple parts that yields an efficient algorithm for finding L-Tromino covers. This tiling technique is a modification of the proof of Theorem \ref{the:Aztec-180} for tiling the Horiyama \emph{et al.} \cite{HIN17} gadgets with I-Trominoes to tiling general regions with L-Trominoes. Finally, we close this paper with a simple lower bound on the number of tromino tilings of an Aztec diamond.

\section{Preliminaries}\label{sec:preliminaries}

In this work we will use $\integer$ to denote the set of integers and $[a,b]$ to denote the discrete interval $\{a,a+1,\dots, b\}$. A region $R$ is a finite union of cells, such that the interior is connected.  If a cell is the set of points $[a,a+1]\times [b,b+1]$, we label such a cell by $(a,b)$ which we refer to as the \emph{cell's coordinate}. Two cells are adjacent if the Manhattan distance, i.e., the $L_1$-norm, of their coordinates is 1; thus, two cells in diagonal to each other are not adjacent. 

A \emph{tromino} is a polyomino of 3 cells. In general, there are two types of trominoes, the L-tromino and the I-tromino. An L-tromino is a polyomino of 3 cells with an L shape. An I-tromino is a polyomino of 3 straight cells with the form of an I. In this work we will mostly be dealing  with L-Trominos and we will refer to them simply as trominoes; I-trominoes will appear later but we will make sure to clarify which type of tromino we are referring to.

A \emph{defect} is a cell that is ``marked'' in the sense that no tromino can be placed on top of that cell. A \emph{cover} or \emph{tiling} of a region $R$ is a set of trominoes covering all cells of $R$ that are not defects with no overlapping and each tromino is packed inside $R$. The \emph{size} of a cover is the number of tiles in it.

\begin{definition}
$\trio$ is the following problem:

\begin{tabular}{ll}
INPUT	&: a region $R$ with defects.\\
OUTPUT	&: ``yes'' if $R$ has a cover and ``no'' otherwise.
\end{tabular}
\end{definition}

Moore and Robson \cite{MR01} proved that $\trio$ is NP-complete and Horiyama \emph{et al.} \cite{HIN17} proved that $\#\trio$, the counting version of $\trio$, is \#P-complete.

In this work we will also consider tilings where only trominoes with $180^\circ$ rotations are used. More precisely, given a region $R$ we want to find a cover where all trominoes are \emph{right-oriented} as in Figure \ref{fig:180tromino}(a) or \emph{left-oriented} as in Figure \ref{fig:180tromino}(b). We will refer to trominoes where only their $180^\circ$ rotations are considered as \emph{180-trominoes}. A \emph{180-cover} of $R$ is a cover with 180-trominoes.

\begin{figure}[htb!]
\centering
\subfloat[Right-oriented]{
\begin{tikzpicture}[scale=0.8]
	\draw (1,1) rectangle (2,2);
	\draw (0,0) rectangle (1,1);
	\draw (1,0) rectangle (2,1);
\end{tikzpicture}
\quad
\begin{tikzpicture}[scale=0.8]
	\draw (1,1) rectangle (2,2);
	\draw (0,0) rectangle (1,1);
	\draw (0,1) rectangle (1,2);
\end{tikzpicture}
}
\qquad\qquad
\subfloat[Left-oriented]{
\begin{tikzpicture}[scale=0.8]
	\draw (1,1) rectangle (2,2);
	\draw (0,1) rectangle (1,2);
	\draw (1,0) rectangle (2,1);
\end{tikzpicture}
\quad
\begin{tikzpicture}[scale=0.8]
	\draw (0,0) rectangle (1,1);
	\draw (0,1) rectangle (1,2);
	\draw (1,0) rectangle (2,1);
\end{tikzpicture}
}
\caption{The $\strio$ problem either takes trominoes from the left figure or the right figure.}
\label{fig:180tromino}
\end{figure}

\begin{definition}
$\strio$ is the following problem:

\begin{tabular}{ll}
INPUT	&: a region $R$ with defects.\\
OUTPUT	&: ``yes'' if $R$ has a 180-cover and ``no'' otherwise.
\end{tabular}
\end{definition}

\section{Tiling of the Aztec Rectangle}\label{sec:Aztec}
The \emph{Aztec diamond of order $n$}, denoted $\ad(n)$, is the union of lattice squares $[a,a+1]\times [b,b+1]$, with $a,b\in \integer$, that lie completely inside the square $\{(x,y)\> | \> |x|+|y|\leq n+1\}$. Figure \ref{fig:azn} shows the first four Aztec diamonds. Tilings of the Aztec diamond with dominoes was initiated by Elkies \emph{et al.} \cite{EKLP92}, which gave impetus to a lot of current work in this area.

The concept of an Aztec diamond can be very easily extended to that of an \emph{Aztec rectangle}. We denote by $\arr$ the Aztec rectangle 
which has $a$ unit squares on the southwestern side and $b$ unit squares on the northwestern side; in the case when $a=b=n$ we get an Aztec diamond of order $n$---in view of this, all results about Aztec rectangles stated in this paper can also be stated for Aztec diamonds, and therefore, we do not mention all such results as corollaries.

When dealing with Aztec rectangles, with no loss of generality, we always assume that $a<b$.
As an example Figure \ref{fig:ar} shows $\ar_{4,10}$. Domino tilings of Aztec rectangles have been studied by various mathematicians starting with Mills \emph{et. al.} \cite{MRR83}.

In the following subsections we study tilings of the Aztec rectangle using trominoes with and without defects, and then specialize them to Aztec diamonds.
\begin{figure}[htb!]
\centering
\subfloat[$\ad(1)$]{
	\includegraphics[scale=0.022]{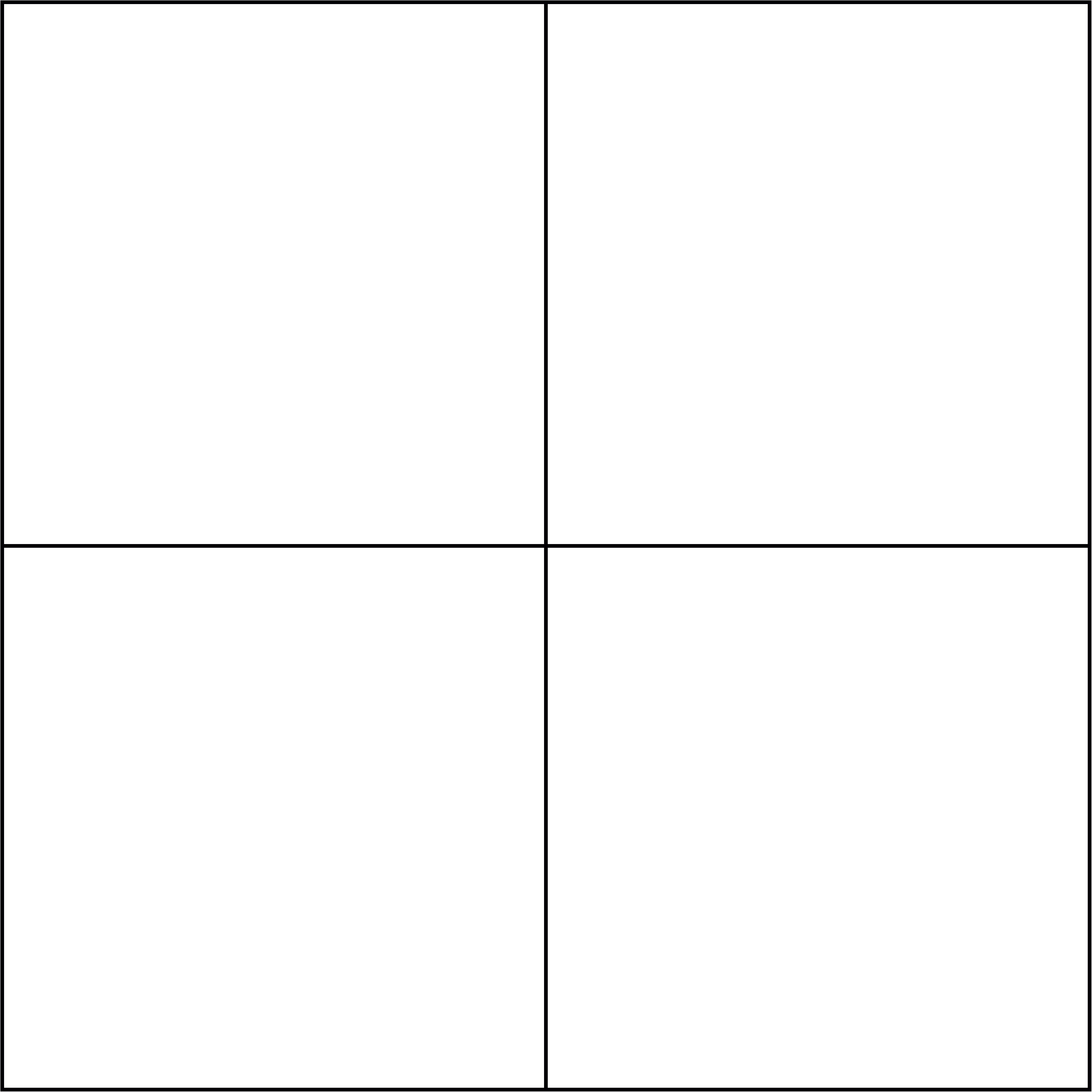}
}
\subfloat[$\ad(2)$]{
	\includegraphics[scale=0.033]{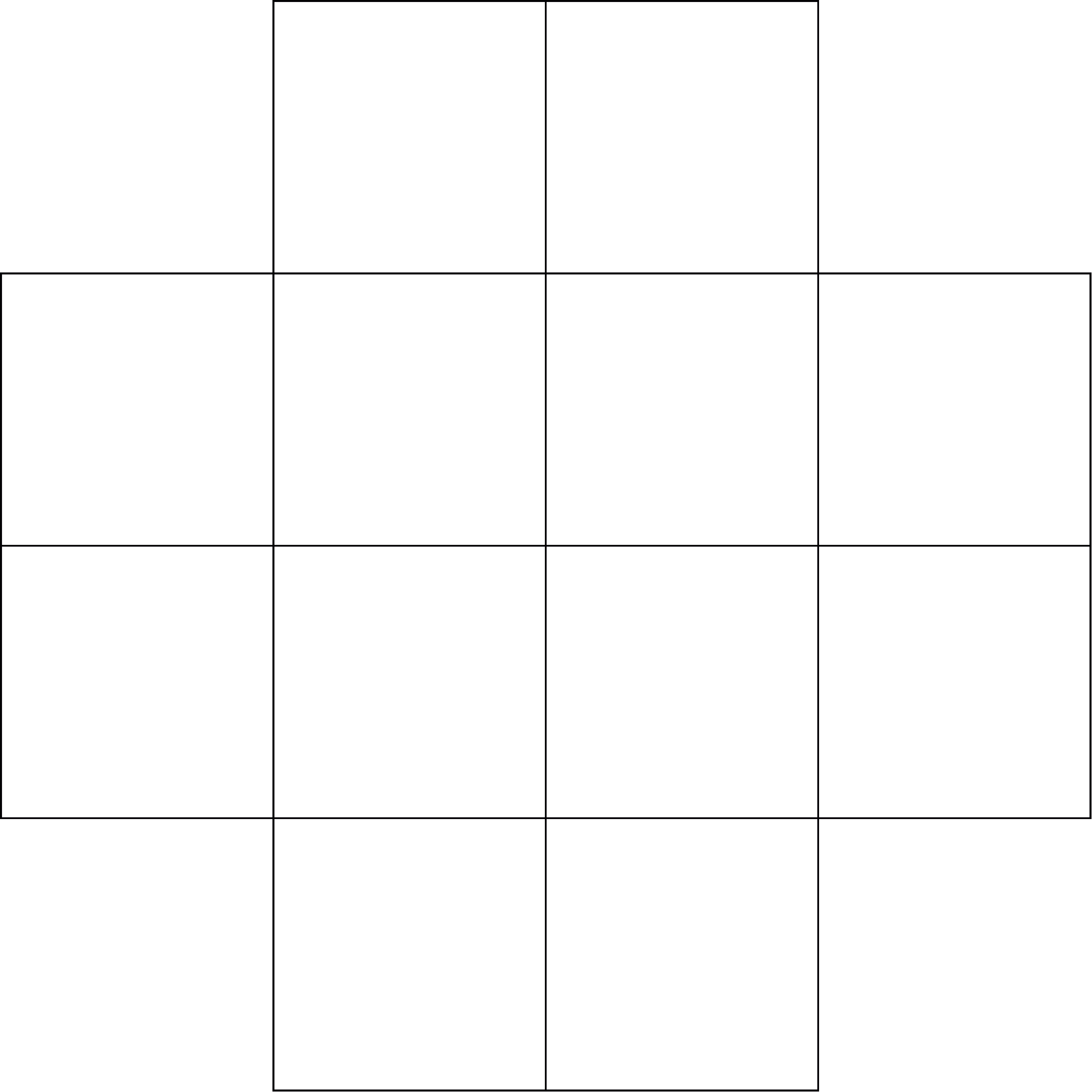}
}
\subfloat[$\ad(3)$]{
	\includegraphics[scale=0.030]{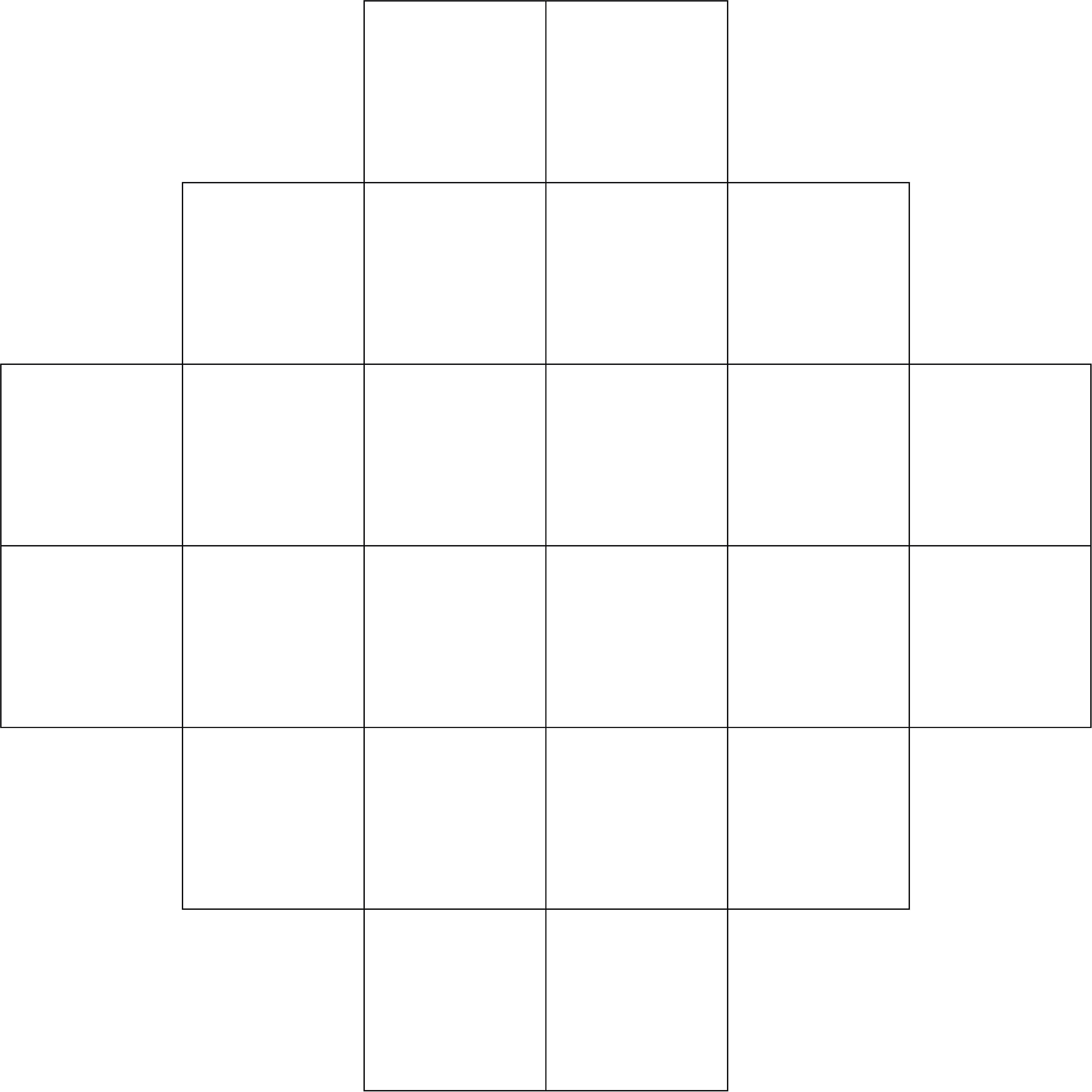}
}
\subfloat[$\ad(4)$]{
	\includegraphics[scale=0.023]{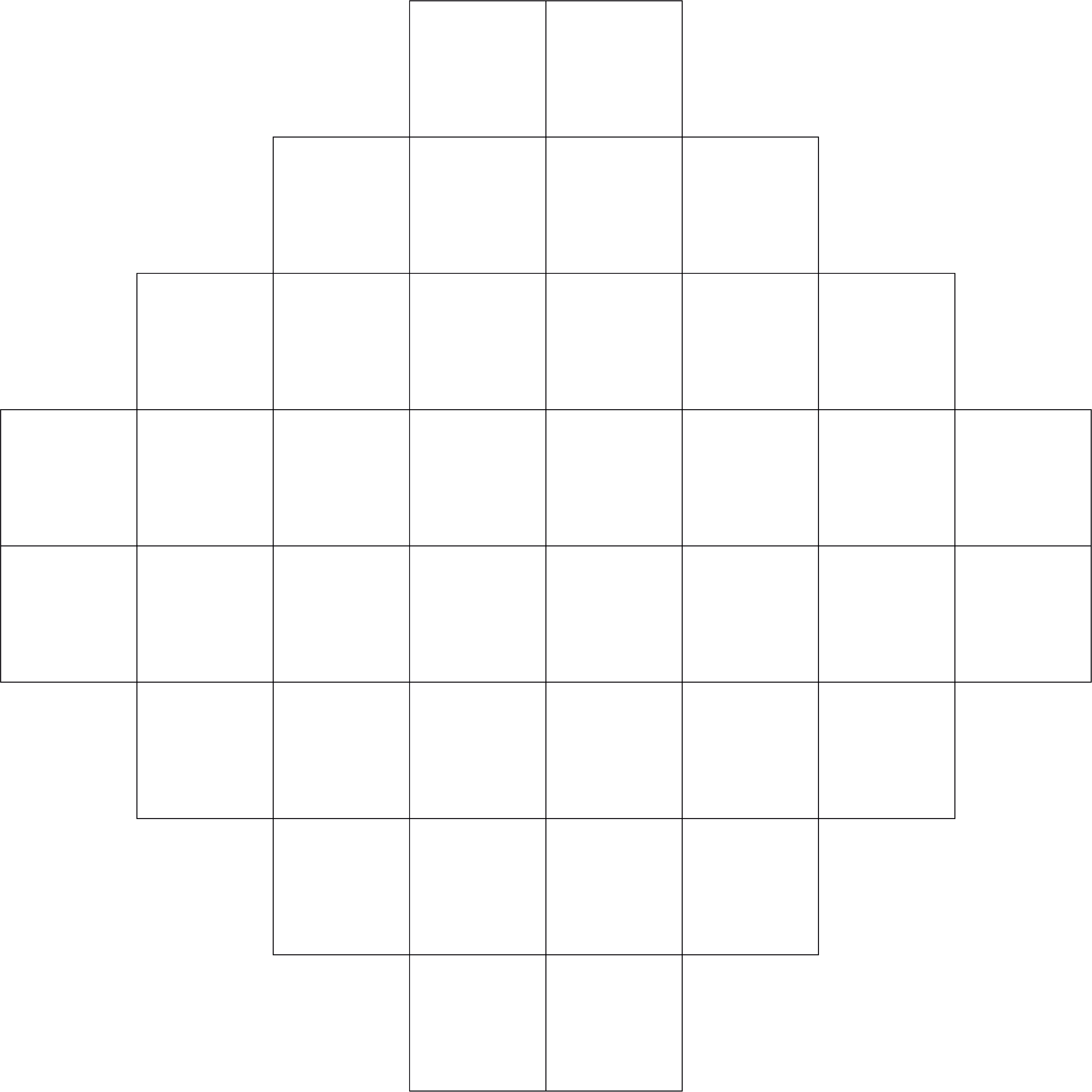}
}
\caption{Aztec diamonds of order 1, 2, 3 and 4.}
\label{fig:azn}
\end{figure}

\begin{figure}[htb!]
	\centering
	\includegraphics[scale=0.5]{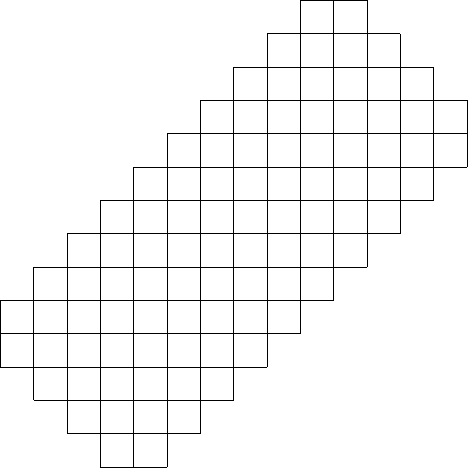}
	\caption{Aztec rectangle $\ar_{4,10}$.}
	\label{fig:ar}
\end{figure}

\subsection{Tilings with No Defects}\label{sec:az-defects}

For any Aztec rectangle $\arr$ with no defects, we can completely understand when there is a tiling. The following theorem gives a characterization.

\begin{theorem}\label{th-1}
 $\ar_{a,b}$ has a cover if and only if $a(b+1)+b(a+1)\equiv 0 \pmod 3$.
\end{theorem}

\noindent As a corollary, we get the following for the Aztec diamond.

\begin{corollary}\label{the:az-tile}
$\ad(n)$ has a cover if and only if $n(n+1)\equiv 0 \pmod 3$.
\end{corollary}

\begin{proof}
Put $a=b=n$ in Theorem \ref{th-1}.
\end{proof}

To prove Theorem \ref{th-1}, first we present tilings of particular cases of the Aztec rectangle in Lemmas \ref{lem-1} and \ref{lem-2}. The following lemma is trivial.

\begin{lemma}\label{squares}
 An Aztec rectangle, $\arr$ contains $a(b+1)+b(a+1)$ unit squares. Further, specializing $a=b=n$ we get that an Aztec diamond of order $n$ contains $2n(n+1)$ unit squares.
\end{lemma}

Define a \emph{stair}  as a polyomino made-up only of trominoes with their $180^\circ$ rotations connected as in Figure \ref{fig:stair}(a). The same stair can be rotated $90^\circ$ to obtain another stair. A \emph{$k$-stair} is a co-joined set of $k$ stairs, where a stair is joined to another stair by matching their extremes; for example, in Figure \ref{fig:stair}(b) we can see two stairs where the lowest extreme of the upper stair is matched with the upper extreme of the lower stair. This idea is easily extended to a set of $k$ stairs thus giving a $k$-stair as in Figure \ref{fig:stair}(c). A $k$-stair can also be rotated $90^\circ$ to obtain another $k$-stair. The \emph{height of a $k$-stair} is the number of steps in it. It is easy to see that the height of a $k$-stair is $3k+2$. In addition, a single tromino would be a $0$-stair.

\begin{figure}[htb!]
\centering
\subfloat[1-stair]{
	\includegraphics[scale=0.1]{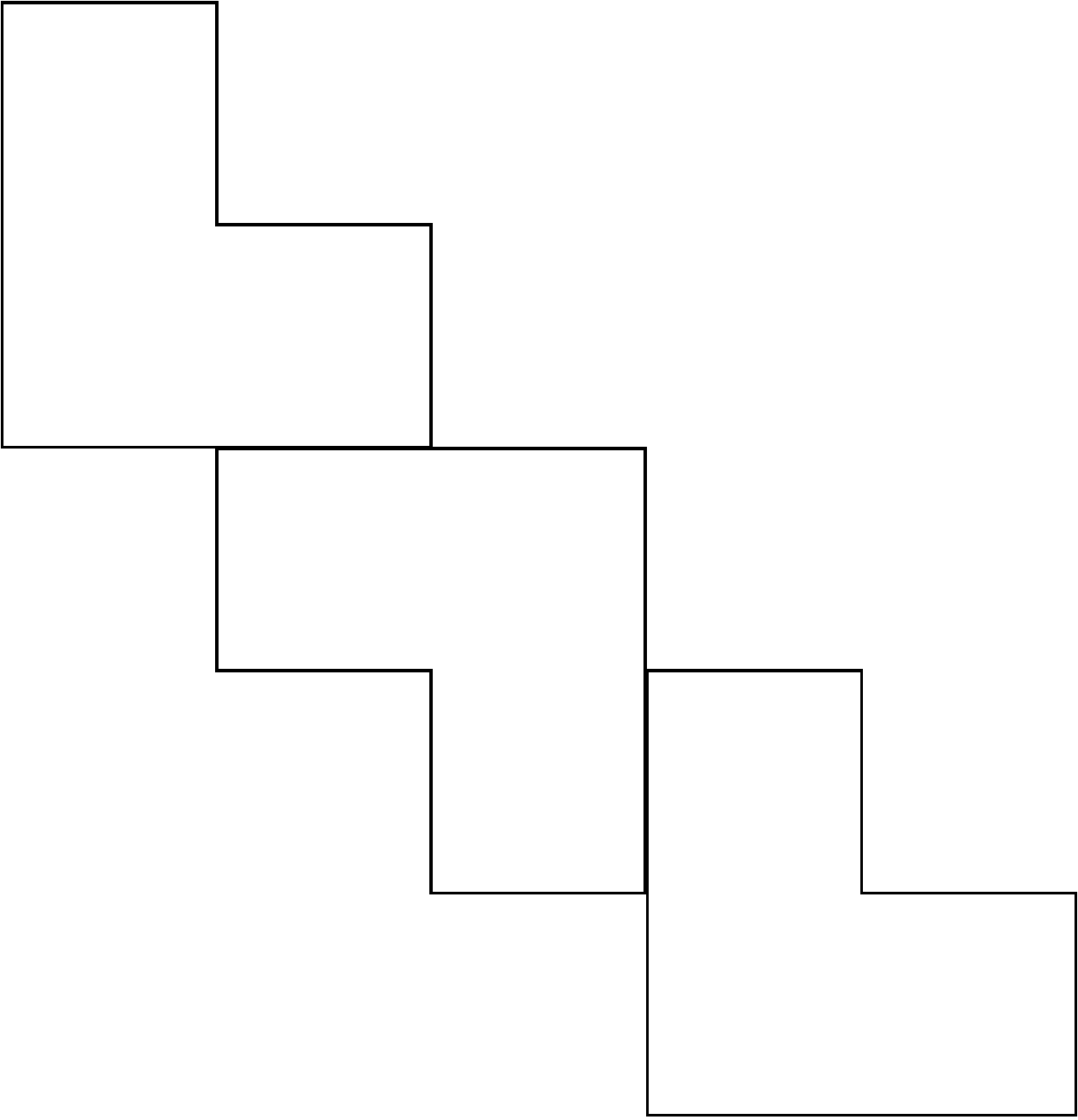}
}
\quad\quad\quad\quad
\subfloat[2-stair]{
	\includegraphics[scale=0.07]{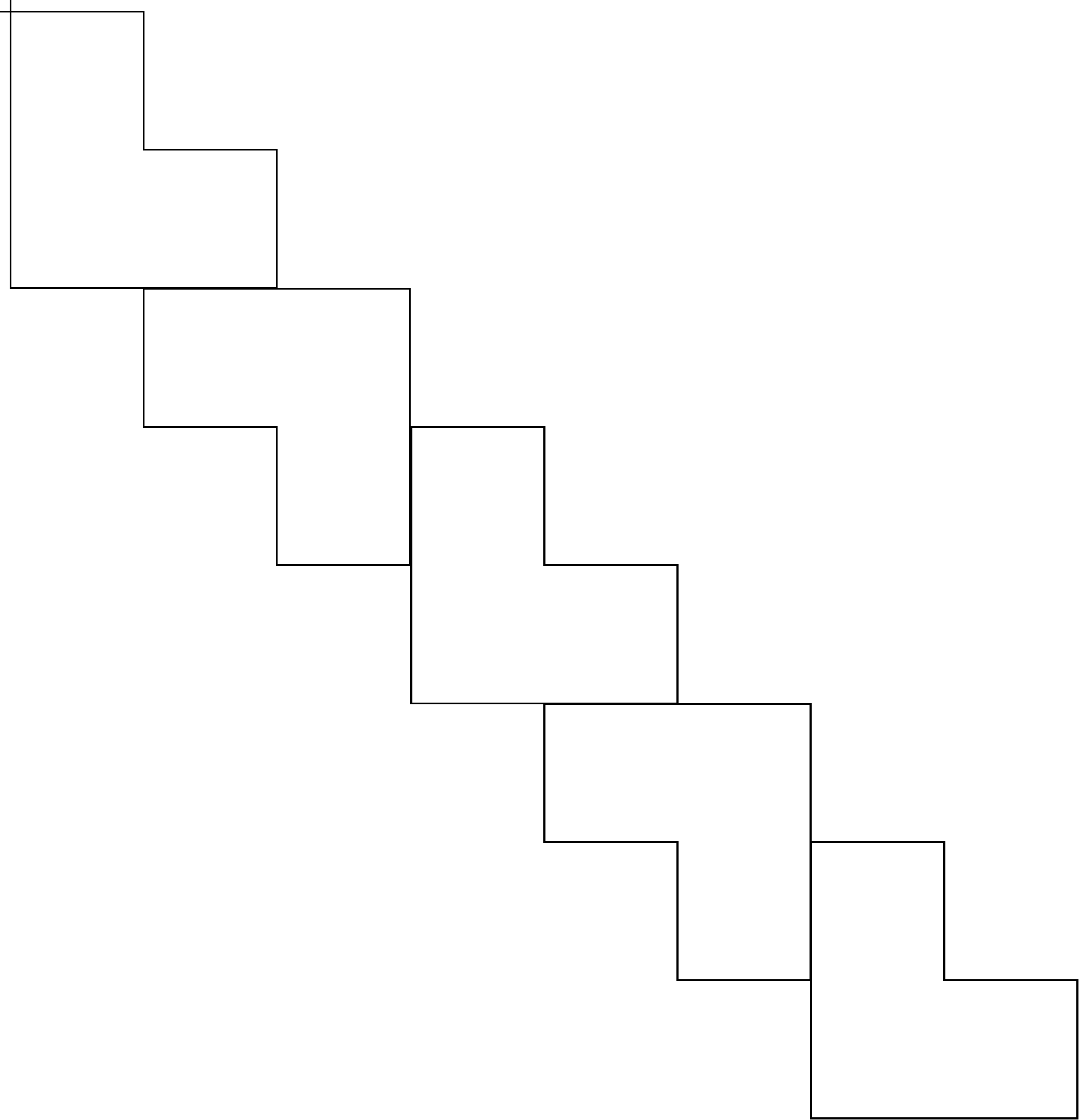}
}
\quad\quad\quad\quad
\subfloat[$k$-stair]{
	\includegraphics[scale=0.06]{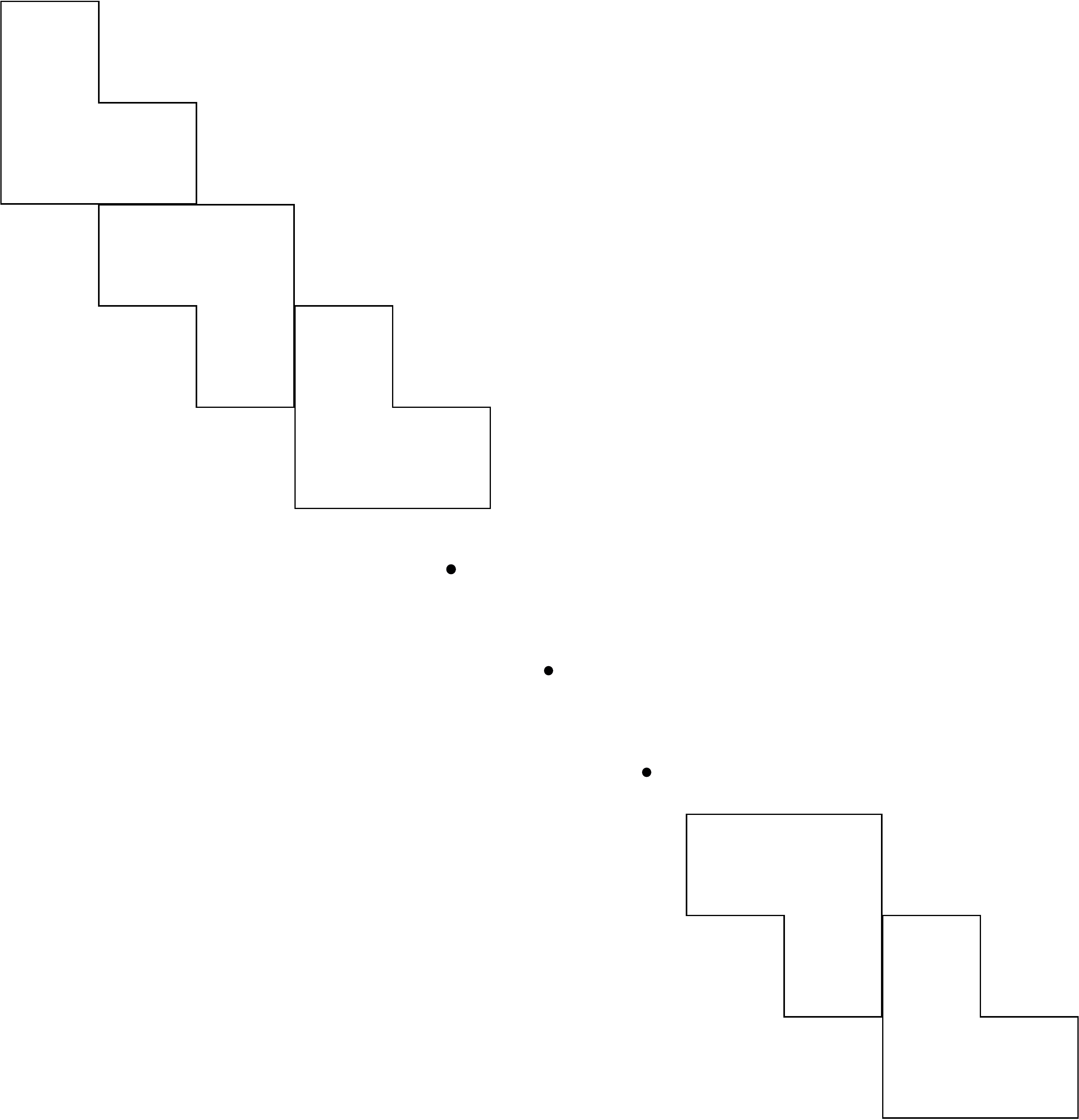}
}
\caption{A stair also includes all $90^\circ$ rotations.}
\label{fig:stair}
\end{figure}

\begin{figure}[htb!]
\centering
\subfloat[Tiling with a single stair.]{
	\includegraphics[scale=0.5]{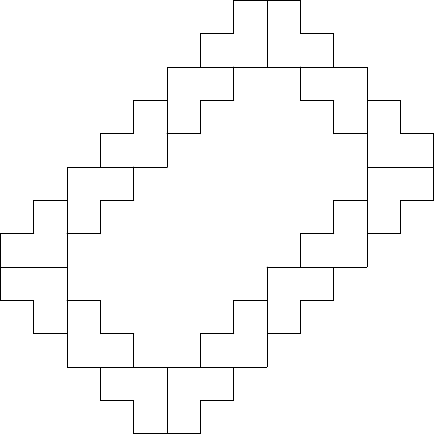}
}
\quad
\subfloat[Tiling with a double stair.]{
	\includegraphics[scale=0.45]{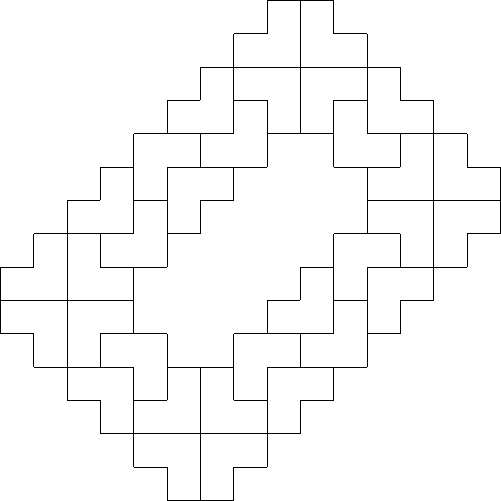}
}
\caption{Tilings of Lemmas \ref{lem-1} and \ref{lem-2}.}
\label{fig:borders}
\end{figure}

\begin{lemma}\label{lem-1}
If $3\mid a,b$ and $\arr$ has a cover, then $\ar_{a+2,b+2}$ has a cover. 
\end{lemma}

\begin{proof}
If $a,b$ are multiples of 3, then an $a/3$-stair and an $b/3$-stair can be used to tile around $\arr$ along the shorter and longer sides 
respectively, using the pattern of Figure \ref{fig:borders}(a). This tiling increments the order of the Aztec rectangle
by 2, thus obtaining a tiling for $\ar_{a+2,b+2}$.
\end{proof}

 \begin{lemma}\label{lem-2}
  If $3\mid a+1,b+1$ and $\arr$ has a cover, then $\ar_{a+4,b+4}$ has a cover.
 \end{lemma}

\begin{proof}
To find a tiling for $\ar_{a+4,b+4}$ we use four copies of $\ad(2)$ added to the four corners of $\arr$. Then, to complete the tiling, we use two 
$(a-2)/3$ and $(b-2)/3$-stairs one on top of each other along the shorter and longer sides respectively, to complete the border. The entire construction follows the pattern of Figure \ref{fig:borders}(b). This tiling increments the order of the Aztec rectangle by 4, thus obtaining a tiling for $\ar_{a+4,b+4}$.
\end{proof}

Now, let us prove Theorem \ref{th-1}.

\begin{proof}[Proof of Theorem \ref{th-1}]
The values for which $a(b+1)+b(a+1)\equiv 0 \pmod 3$ holds are $a,b=3k$ and $a,b=3k-1$ for some $k\in \mathbb{Z}$. 

Thus, the statement is equivalent to saying that for all positive integers $k$ there is a tiling of $\arr$ where $3\mid a,b$  or $3\mid a+1,b+1$ and that there are no tilings for $\arr$ when $3\mid a+2,b+2$.

 We show the second part now, which is easy since if we have $\arr$ with $a,b$ of the form $3k+2$, then the number of lattice 
 squares inside $\arr$ is not divisible by $3$ and hence we cannot tile this region with trominoes.

We come to the first cases now. Using Lemmas \ref{lem-1} and \ref{lem-2}, this part is clear if we can show the base induction case to be true.

The base case of Lemma \ref{lem-1} is shown in Figure \ref{fig:base-1}(a), which is $\ar_{3,6}$. Once we have a tiling of $\ar_{3,6}$, we 
can use Lemma \ref{lem-1} to create a tiling of an Aztec rectangle whose sides are increased by $2$. We can also increase 
$\ar_{3,6}$ by using the additional pieces shown in Figure \ref{fig:base-1}(b,c) using them in combinations with any case of Aztec rectangle 
satisfying the properties of Lemma \ref{lem-1} to increase either the longer or the shorter sides, and if all three additional pieces are 
used then we can increase both sides of $\arr$.

\begin{figure}[htb!]
\centering
\subfloat[Base induction case.]{
	\includegraphics[scale=0.5]{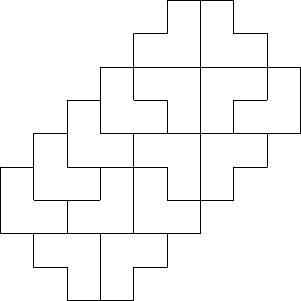}
}
\quad
\subfloat[Length additional pieces.]{
	\includegraphics[scale=0.45]{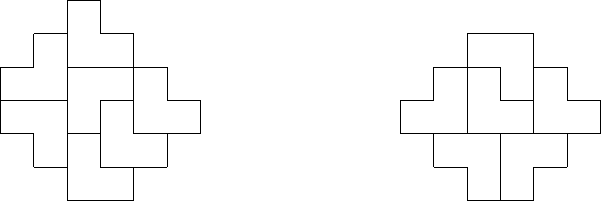}
}
\quad
\subfloat[Breadth additional piece.]{
	\includegraphics[scale=0.45]{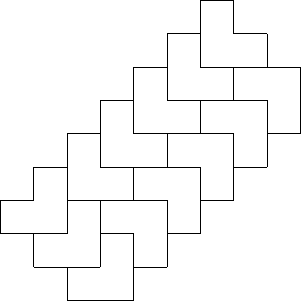}
}
\caption{Base case of Lemma \ref{lem-1}.}
\label{fig:base-1}
\end{figure}

Similarly, the base case of Lemma \ref{lem-2} is shown in Figure \ref{fig:base-2}(a), which is $\ar_{2,5}$. Once we have a tiling of $\ar_{2,5}$, we 
can use Lemma \ref{lem-2} to create a tiling of an Aztec rectangle whose sides are increased by $4$. We can also increase 
$\arr$ by using the additional pieces shown in Figure \ref{fig:base-2}(b,c,d) using them in combinations with any case of Aztec rectangle 
satisfying the properties of Lemma \ref{lem-2} to increase either the longer or the shorter sides, and if all three additional pieces are 
used then we can increase both sides of $\arr$.
\end{proof}

\begin{figure}[htb!]
\centering
\subfloat[Base induction case.]{
	\includegraphics[scale=0.45]{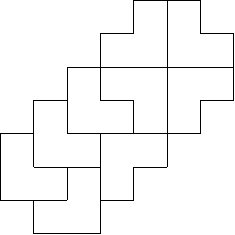}
}
\quad
\subfloat[Length additional piece.]{
	\includegraphics[scale=0.45]{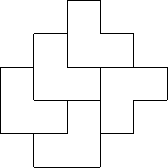}
}
\quad
\subfloat[Length additional piece.]{
	\includegraphics[scale=0.45]{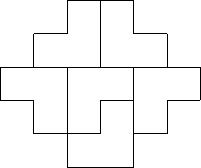}
}
\quad
\subfloat[Breadth additional piece.]{
	\includegraphics[scale=0.45]{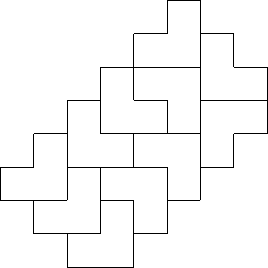}
}
\caption{Base case of Lemma \ref{lem-2}.}
\label{fig:base-2}
\end{figure}

An $O(b^2)$ time algorithm is immediately obtained from the proof of Theorem \ref{th-1}.

\begin{theorem}\label{the:algorithm}
A tromino cover for $\arr$ can be found in time $O(b^2)$. 
\end{theorem}
\begin{proof}
Given $a,b\in \nat$, the following procedure called $\texttt{ARTiling}(a,b)$ finds a tiling for $\arr$.
\begin{enumerate}
	\item If $a=2,b=5$ or $a=3,b=6$, return the tiling of Figure \ref{fig:base-2}(a) or Figure \ref{fig:base-1}(b), respectively.
	\item If $a(b+1)+b(a+1)\not\equiv 0 \pmod 3$, then return ``there is no tiling''.\label{st:no-tiling}
	\item If $a,b$ are multiples of 3, then\label{st:multiple}
		\begin{enumerate}
		\item $R\gets \texttt{ARTiling}(a-2,b-2)$;
		\item fill the borders of $R$ using the pattern of Figure \ref{fig:base-1}.\label{st:border4}
		\end{enumerate}
	\item If $a+1,b+1$ is a multiple of 3, then
		\begin{enumerate}
		\item $R\gets \texttt{ARTiling}(a-4,b-4)$;
		\item fill the borders of $R$ using the pattern of Figure \ref{fig:base-2}.\label{st:border2}
		\end{enumerate}
	\item Return $R$.
\end{enumerate}
Steps \ref{st:no-tiling} and \ref{st:multiple} are done in time $O(\log b)$ and steps \ref{st:border4} and \ref{st:border2} can be done in time $O(b)$, thus, giving a total time complexity of $O(b^2)$. The correctness follows from Lemmas \ref{lem-1} and \ref{lem-2}.
\end{proof}

\begin{corollary}
A tromino cover for $\ad(n)$ can be found in time $O(n^2)$. 
\end{corollary}

\subsection{Tiling with Defects}

From Theorem \ref{th-1} we know that for any positive integers $a,b$, the Aztec rectangles with no defects $\arr$ such that $3$ divides 
$a,b$ or $3$ divides $a+1,b+1$ have a cover but if $3$ divides $a+2,b+2$, then $\arr$ does not have a tiling.  We show that if such an Aztec rectangle has exactly one defect, then it can be covered with trominoes.

\begin{theorem}\label{the:ar-defect}
 There exists a cover for $\ar_{a,b}$ with $a,b$ of the form $3k-2$ having one defect.
\end{theorem}

\begin{proof}
To tile $\arr$ with one defect we use a construct which we call a \emph{fringe} appearing in Figure \ref{fig:az-defect}(a). It is easy to check that if a fringe has exactly one defect, then it can be covered with trominoes.

To construct a tiling for $\arr$ with one defect we place a fringe in a way that includes the defect and the left and right ends of the fringe touches the boundaries of the Aztec rectangle as in Figure \ref{fig:az-defect}(b). Then we use the tiling pattern of Figure \ref{fig:az-defect}(b) where we put stairs above and below the fringe.
\end{proof}

\begin{figure}[htb!]
\centering
\subfloat[Fringe]{
	\includegraphics[scale=0.1]{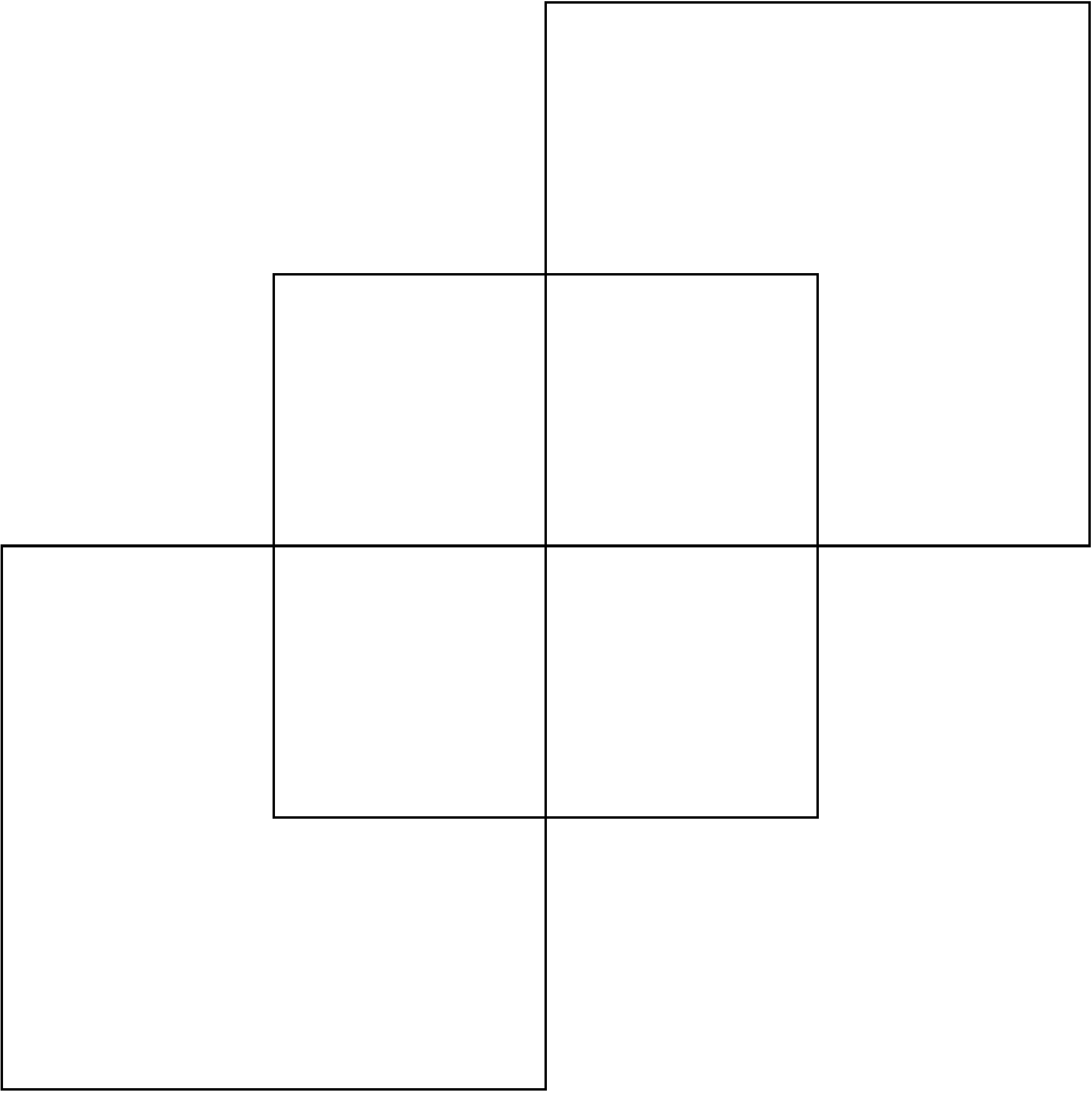}
}
\quad\quad\quad\quad
\subfloat[Tiling pattern]{
	\includegraphics[scale=0.5]{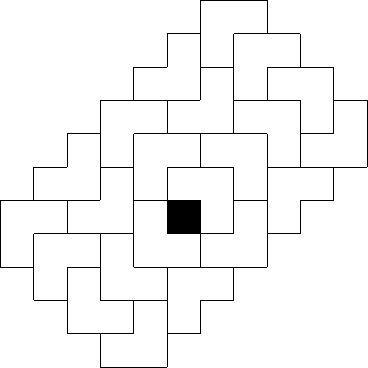}
}
\caption{Tiling of $\arr$ with one defect. A \emph{fringe} can be composed of any number of order 1 Aztec diamonds $\ad(1)$ joined by their upper right and lower left cells. A \emph{reversed fringe} is obtained by joining order 1 Aztec diamonds by their upper left and lower right cells.}
\label{fig:az-defect}
\end{figure}

\begin{corollary}\label{the:az-defect}
For any positive integer $k$, the Aztec diamond $\ad(3k-2)$ with one defect has a cover.
\end{corollary}

We can consider many different classes of defects, and it can be observed that some of these classes have easy tilings. As an example, 
we have in Figure \ref{fig:four-d}(a) an Aztec rectangle with four defects on its corners. A tiling of this region is shown in 
Figure \ref{fig:four-d}(b). In the combinatorics literature, tilings of regions with defects of several kinds for Aztec rectangle have been studied (see 
\cite{MPS17} for the most general class of boundary defects).

\begin{figure}[htb!]
\centering
\subfloat[$\arr$ with four defects]{
	\includegraphics[scale=0.5]{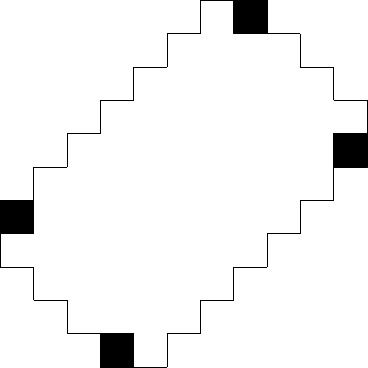}
}
\quad\quad\quad\quad
\subfloat[Tiling pattern]{
	\includegraphics[scale=0.5]{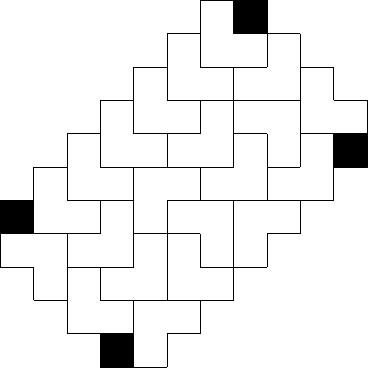}
}
\caption{Tiling of $\arr$ with four defects.}
\label{fig:four-d}
\end{figure}

\begin{remark}
Similar defects can be studied for Aztec diamonds as well. In fact, we can delete all cells in a fringe and obtain a tiling.
\end{remark}

Even though the proof of Theorem \ref{the:ar-defect} gives an $O(b^2)$ time algorithm for finding a cover for $\arr$ with one defect, in general, however, it is computationally hard to determine if $\arr$ with an unknown number of defects has a cover.

\begin{theorem}\label{the:az-defect-hard}
It is NP-complete to decide whether $\ad(n)$ with an unbounded number of defects has a cover.
\end{theorem}
\begin{proof}[Proof Sketch]
The reduction is from tiling an arbitrary region $R^\prime$ with defects. The idea is to embed $R^\prime$ into $\ad(n)$ for some sufficiently large $n$ and insert defects in $\ad(n)$ in a way that surrounds $R^\prime$ (see Figure \ref{fig:embed}). 
\end{proof}

\begin{figure}
    \centering
    \includegraphics[scale=1.2]{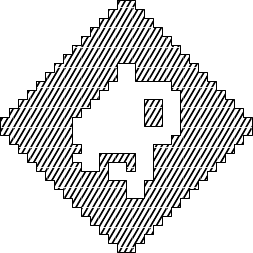}
    \caption{Tiling an Aztec Rectangle with arbitrary defects.}
    \label{fig:embed}
\end{figure}

\begin{corollary}
It is NP-complete to decide if $\arr$ with an unbounded number of defects has a cover.
\end{corollary}

\section{Tiling with 180-Trominoes}\label{sec:180tromino}

In this section we study tilings of arbitrary regions using only 180-trominoes. With no loss of generality, we will only consider right-oriented 180-trominoes.

\begin{figure}[htb!]
	\centering
	\includegraphics[scale=1.4]{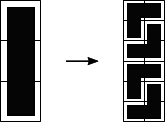}
	\caption{I-Tromino to L-Tromino transformation using 180-Trominoes.}
	\label{fig:itromino-180ltromino}
\end{figure}

\subsection{Hardness}\label{sub:hardness}
It is easy to see that even when restricted to 180-trominoes, deciding the existence of a tiling of an arbitrary region is still hard.
\begin{theorem}\label{the:180tromino}
$\strio$ is NP-complete.
\end{theorem}

\begin{proof}
The proof uses the same gadgets for the reduction for I-Trominoes from the \emph{1-in-3 Graph Orientation Problem} of Horiyama \emph{et al.} \cite{HIN17}. Take any gadget of Horiyama \emph{et al.} \cite{HIN17} and partition each cell into 4 new cells. Thus, each I-tromino is transformed in a new $2\times 6$ or $6\times 2$ region (depending on the orientation of the I-tromino) which can be covered with four 180-trominoes as in Fig \ref{fig:itromino-180ltromino}. If a gadget is covered with I-trominoes, then the same gadget, after partitioning each cell into four new cells, can also be covered with 180-trominoes. To see the other direction of this implication, we exhaustively examined all possible ways to cover each 4-cell-divided gadget with L-trominoes, and observed that each gadget with its original cells can also be covered with I-trominoes. Figures \ref{fig:dup}, \ref{fig:cross} and \ref{fig:clause} show this.
\end{proof}

\begin{figure}[htb!]
\centering
\subfloat[]{
	\includegraphics[scale=0.6]{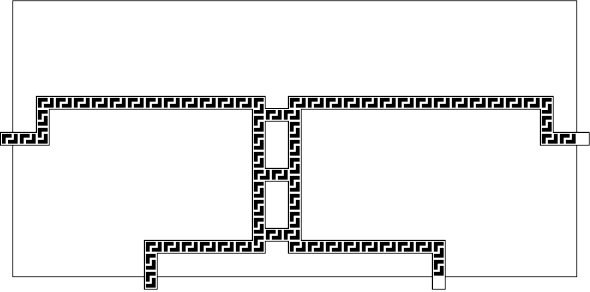}
}
\quad\quad\quad\quad
\subfloat[]{
	\includegraphics[scale=0.6]{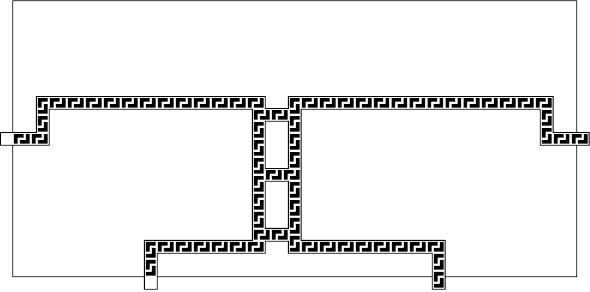}
}
\caption{Tiling of the duplicator gadget.}
\label{fig:dup}
\end{figure}

\begin{figure}[!htb]
\centering
\subfloat[]{
	\includegraphics[scale=0.65]{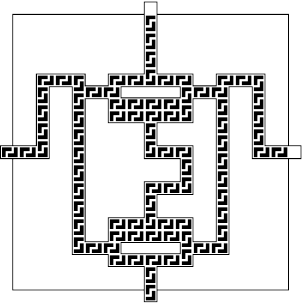}
}
\quad\quad\quad\quad
\subfloat[]{
	\includegraphics[scale=0.65]{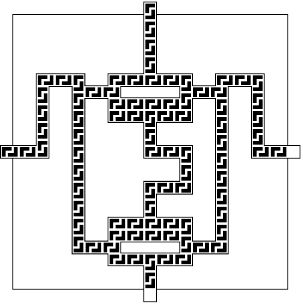}
}\\
\subfloat[]{
	\includegraphics[scale=0.65]{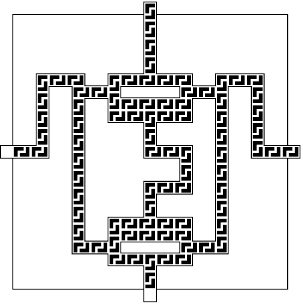}
}
\quad\quad\quad\quad
\subfloat[]{
	\includegraphics[scale=0.65]{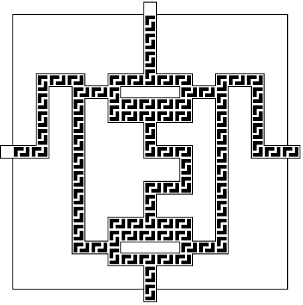}
}
\caption{Tiling of the cross gadget.}
\label{fig:cross}
\end{figure}

\begin{figure}[!htb]
\centering
\subfloat[]{
	\includegraphics[scale=0.65]{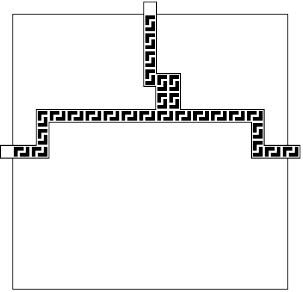}
}
\quad\quad\quad\quad
\subfloat[]{
	\includegraphics[scale=0.65]{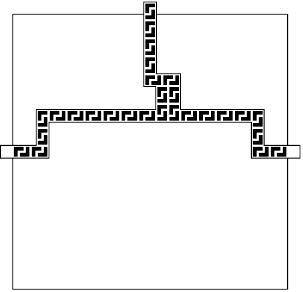}
}
\quad\quad\quad\quad
\subfloat[]{
	\includegraphics[scale=0.65]{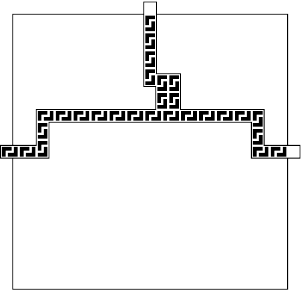}
}
\caption{Tiling of the clause gadget.}
\label{fig:clause}
\end{figure}

\begin{remark}
Theorem \ref{the:180tromino} also implies that the \emph{Triangular Trihex Tiling Problem} of Conway and Lagarias \cite{CL90} (see Figure \ref{fig:trihex}) is NP-complete (see Figure \ref{fig:trihex-2}).

\begin{figure}[htb!]
	\centering
	\includegraphics[scale=0.7]{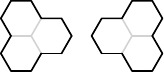}
	\caption{Two triangular trihex.}
	\label{fig:trihex}
\end{figure}
\begin{figure}[htb!]
	\centering
	\includegraphics[scale=0.8]{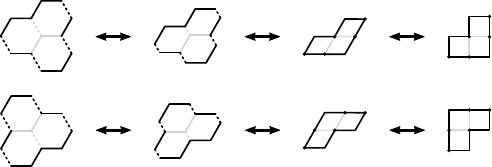}
	\caption{Transformation from triangular trihex to $180$-tromino.}
		\label{fig:trihex-2}
\end{figure}
\end{remark}

It is natural to think along these lines about tiling the Aztec rectangle with 180-trominoes. However, we show that it is impossible.

\begin{theorem}\label{the:Aztec-180}
$\arr$ does not have a 180-cover.
\end{theorem}
\begin{proof}
Consider the southwestern side of any Aztec rectangle as in Figure \ref{fig:az-180tromino} and pick any one of the marked cells, say the cell at coordinate
$(c,d)$. There are only two ways to cover that cell with a right-oriented tromino. With one tromino we can cover the cells 
with coordinates $(c,d),(c,d+1)$ and $(c+1,d+1)$, whereas with the other tromino we can cover the cells $(c,d), (c+1,d)$ and
$(c+1,d+1)$. In either case the cells at $(c,d)$ and $(c+1,d+1)$ are always covered, and depending on which tromino is 
chosen either the cell at $(c,d+1)$ or $(c+1,d)$ is covered. Therefore, if we cover the entire bottom-left side of an Aztec
rectangle, there will always be a cell at $(c,d+1)$ or $(c+1,d)$ that cannot be covered. Note that any reversed fringe 
that is on top of the bottom-left side of any Aztec rectangle can be covered with 180-trominoes if it has one defect.
\end{proof}

\begin{corollary}
$\ad(n)$ does not have a 180-cover.
\end{corollary}

\begin{figure}[htb!]
\centering
\includegraphics[scale=0.55]{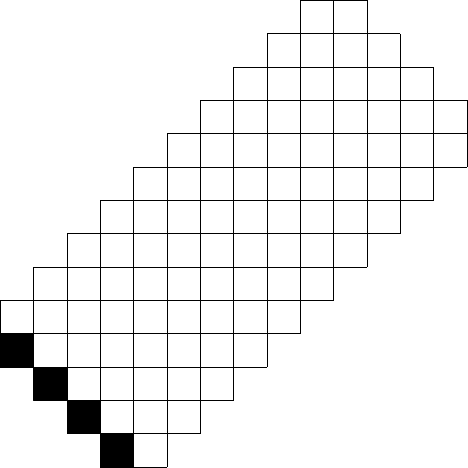}
\caption{Covering of an Aztec rectangle with right-oriented trominoes.}
\label{fig:az-180tromino}
\end{figure}

\subsection{Efficient Tilings}\label{sub:efficient}
In this subsection we identify a sufficient condition for a region to have an efficient algorithm that decides the existence of a 180-cover.

\begin{theorem}\label{the:forbidden}
If a region $R$ does not contain any of the forbidden polyominoes of Figure \ref{fig:forbidden} as a subregion, then there exists a polynomial-time algorithm that decides whether $R$ has a 180-cover.
\end{theorem}

\begin{figure}[htb!]
	\centering
	\includegraphics[scale=0.8]{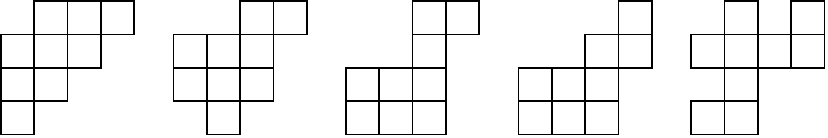}
	\caption{Forbidden polyominoes. All $180^\circ$ rotations, reflections and shear transformations are also forbidden polyominoes.}\label{fig:forbidden}
\end{figure}

A natural and practical application of Theorem \ref{the:forbidden} is in determining if a region does not contain as a subregion any of the forbidden polyominoes, and then, run a polynomial-time algorithm to find a 180-cover. Since the number of forbidden polyominoes is constant, we can check for their non-existence in linear-time. Also, a more theoretical application of Theorem \ref{the:forbidden} is that it could serve as a stepping-stone towards a characterization of regions admitting polynomial-time algorithms.

For the remainder of this subsection we present a proof of Theorem \ref{the:forbidden}. Remember that, with no loss of generality, we only consider right-oriented trominoes. Given a region $R$ we construct a graph $G_R$, which we call the \emph{region graph} of $R$, as follows. For each cell $(a,b)$ that is not a defect there is a vertex $v_{ab}$. There is an edge for each pair of adjacent cells and for each pair $v_{ab}$ and $v_{(a+1)(b+1)}$. Note that this reduction is one-to-one. We present an example in Figure \ref{fig:graphs}.

\begin{figure}[htb!]
	\centering
	\subfloat[Region $R$]{
	\centering
	\includegraphics[scale=1.7]{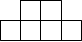}
	}
	\quad\quad\quad\quad\quad\quad
	\subfloat[Region graph $G_R$]{
	\centering
	\includegraphics[scale=1.7]{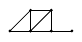}
	}
	\caption{Example of a region graph construction.}
	\label{fig:graphs}
\end{figure}

From the region graph $G_R$ we construct a new graph $I_R$ which we call an \emph{intersection graph} as follows. For each triangle in $G_R$ there is a vertex $t$ and there is an edge between vertices $t_i$ and $t_j$ if the corresponding triangles share a vertex in $G_R$; for example, the intersection graph of Figure \ref{fig:graphs} is a triangle, because all triangles in the region graph share at least one vertex.

\begin{lemma}\label{lem:ind}
For any region $R$ with a fixed number of defects, the maximum number of 180-trominoes that fit in $R$ equals the size of a maximum independent set in $I_R$.
\end{lemma}
\begin{proof}
Let $k$ be the maximum number of tiles that fit in $R$ and let $S$ be a maximum independent set in the intersection graph $I_R$. We claim that $|S|=k$.

Each triangle in the region graph $G_R$ correspond to a position where a 180-tile can fit. If $k$ is the maximum number tiles that can fit in $R$, then there exists a set of $k$ triangles in $G_R$, denoted $T$, that do not share any common vertex. Each triangle in $T$ corresponds to a vertex in $I_R$ and since none of the triangles in $T$ share a common vertex, $T$ defines an independent set in $I_R$ and $k\leq |S|$. 

To prove that $|S|=k$ suppose by contradiction that $T$ is not a maximum independent set of $I_R$, that is, $k<|S|$. Since $S$ is an independent set in $I_R$, there are $|S|$ triangles in $G_R$ that do not share a common vertex. Thus, we can fit $|S|$ 180-trominoes in $R$, which is a contradiction because $k<|S|$.
\end{proof}

The idea for a proof of Theorem \ref{the:forbidden} is to construct a polynomial time algorithm that decides the existence of a 180-cover by deciding if a maximum independent set in $I_R$ equals the number of cells of $R$ divided by 3, which agrees with the number of trominoes covering $R$. Deciding the existence of a maximum independent set of a given size is a well-known NP-complete problem, nevertheless, it is known from the works of Minty \cite{Min80}, Sbihi \cite{Sbi80} and Nakamura and Tamura \cite{NT01} that for claw-free graphs\footnote{A graph is \emph{claw-free} if it does not have $K_{1,3}$ (a claw) as an induced subgraph.} finding independent sets can be done in polynomial time. Hence, if $I_R$ is claw-free, then we can use a polynomial time algorithm for finding independent sets to decide the existence of a 180-cover. If $I_R$ has a claw, however, each claw will give one of the forbidden polyominoes.

In Lemma \ref{lem:reduction} below we show that $\strio$ is polynomial time reducible to deciding independent sets, which allow us to construct algorithms for $\strio$ using known algorithms for deciding independent sets. Then in Lemma \ref{lem:claw} we show that if $I_R$ has a claw, then that claw corresponds to a forbidden polyomino in the region $R$.

\begin{lemma}\label{lem:reduction}
There is a many-one polynomial-time reduction from $\strio$  to the problem of deciding existence of an independent set of a given size.
\end{lemma}
\begin{proof}
First the reduction constructs the region graph $G_R$ and the intersection graph $I_R$. If the size of the largest independent set equals the number of cells of $R$ divided by 3, then output ``yes'' because $R$ has a 180-cover; otherwise output ``no'' because $R$ does not have a 180-cover.

Suppose $R$ has a 180-cover. If $n$ is the number of cells in $R$, then the number of tiles in the 180-cover is $n/3$. By Lemma \ref{lem:ind}, the largest independent set in $I_R$ equals $n/3$. 

Now suppose $R$ does not have a 180-cover. If $n$ is the number of cells in $R$, then $n/3$ is not equal the maximum number of tiles that can fit in $R$. Thus, by Lemma \ref{lem:ind}, it holds that $n/3$ is not equal the size of the largest independent set in $I_R$.
\end{proof} 

\begin{lemma}\label{lem:claw}
If $I_R$ has a claw, then $R$ has at least one forbidden polyomino.
\end{lemma}
\begin{proof}
For any claw in $I_R$ there is a vertex of degree 3 and three vertices of degree 1, and each vertex in $I_R$ corresponds to a triangle in the region graph $G_R$. We refer to the triangle that corresponds to the degree 3 vertex as the \emph{central triangle} and each degree 1 triangle is called an \emph{adjacent triangle}.  Thus, to obtain all forbidden polyominoes, we look at all posible ways to connect (by the vertices) each adjacent triangle to the central triangle in such a way that 
each adjacent triangle only connects to the central triangle in a single vertex and it is not connected to any other adjacent triangle; otherwise, if an adjacent triangle connects with two vertices of the central triangle or any two adjacent vertices connects with one another, then the induced graph does not corresponds to a claw. In Figure \ref{fig:shear} we show two examples on how to obtain a polyomino from another by a shear transformation. By exhaustively enumerating all possibilities as done in Figure \ref{fig:shear}, we can extract all polyominoes that correspond to claws in $I_R$. We partition this set of polyominoes in five equivalence classes, where two polyominoes are in the same class if and only if one can be obtained from the other by a $180^\circ$ rotation, a reflection or shear transformation. All of these equivalence classes can be seen in Figure \ref{fig:app-forbidden}. 
\end{proof}

\begin{figure}
    \centering
    \subfloat[Polyomino 1 transformation.]{
        \centering
        \includegraphics[scale=1.1]{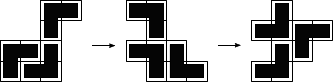}
    }\quad\quad
    \subfloat[Polyomino 2 transformation.]{
        \centering
        \includegraphics[scale=1.1]{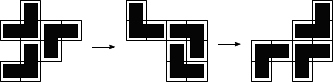}
    }\\
    \subfloat[Region graph of polyomino 1.]{
        \centering
        \includegraphics[scale=1.1]{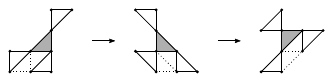}
    }\quad\quad
    \subfloat[Region graph of polyomino 2.]{
        \centering
        \includegraphics[scale=1.1]{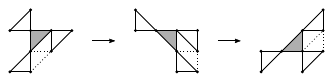}
    }
    \caption{Shear transformation of two polyominoes. The central triangle is painted light gray in each example. Each arrow $\to$ represents a shear transformation.}
    \label{fig:shear}
\end{figure}

\begin{figure}[htb!]
	\centering
	\subfloat{
		\includegraphics[scale=0.5]{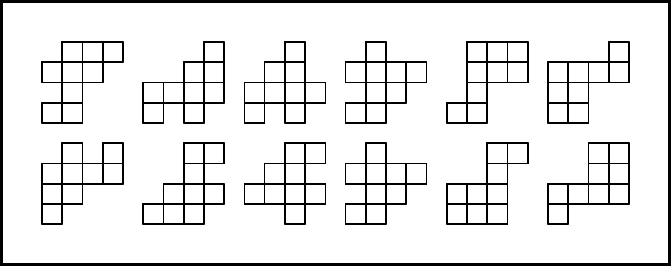}
	}
	\hs{2}
	\subfloat{
		\includegraphics[scale=0.5]{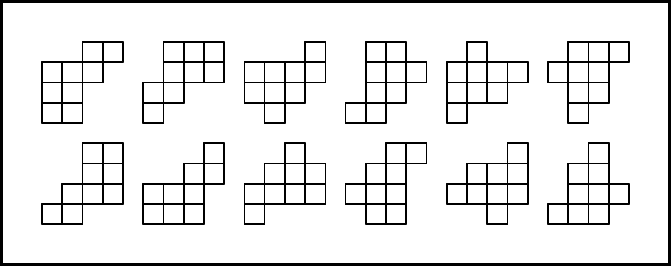}
	}
	\\
	\subfloat{
		\includegraphics[scale=0.5]{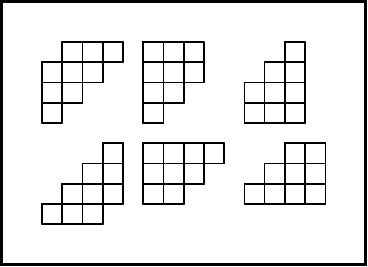}
	}
	\hs{2}
	\subfloat{
		\includegraphics[scale=0.5]{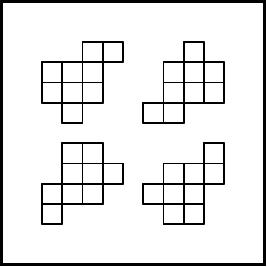}
	}
	\hs{2}
	\subfloat{
		\includegraphics[scale=0.5]{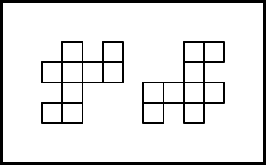}
	}
	\caption{All possible subregions that generate a claw in $I_R$ with their corresponding forbidden polyomino. Any two polyominoes inside a rectangle can be obtained from one to the other by a $180^\circ$ rotation, a reflection or shear transformation.}
	\label{fig:app-forbidden}
\end{figure}

\noindent Lemmas \ref{lem:reduction} and \ref{lem:claw} complete the proof of Theorem \ref{the:forbidden}.

\section{I-Trominoes vs L-Trominoes}\label{sec:li}
In Section \ref{sec:180tromino} we saw that any gadget of Horiyama \emph{et al.} \cite{HIN17} can be covered with I-trominoes if and only if the same gadget, after partitioning each cell into four new cells, can be covered with 180-trominoes. In general, if $R$ is any region and $R^\boxplus$ is the region $R$ where each cell is partitioned into four cells, we have that if $R$ can be covered with I-trominoes, then $R^\boxplus$ can be covered with 180-trominoes.  We do not know, however, if the other way of this implication holds in the general case. The following theorem partly answers this open problem.

\begin{theorem}\label{the:detachable}
Let $R$ be a connected region of size $n$. The region $R^\boxplus$ has an L-Tromino cover if and only if 3 divides $n$.
\end{theorem}

\noindent It is clear that if $R^\boxplus$ has an L-Tromino cover, then 3 divides $n$. We prove the other part in the remainder of this section.

A \emph{planar dual graph}\footnote{Please note that this is an abuse of the term \emph{dual graph} used in graph theory. The practise is, however, standard in combinatorics of tilings.} of a region $R$ is a graph where for each cell $(a,b)$ that is not a defect there is a vertex $v_{ab}$ in $G$ and there is an edge for each pair of adjacent cells. The difference between a planar dual graph defined here and the region graph of Section \ref{sec:180tromino} is that a region graph is a planar dual graph with the addition that there is an edge between each pair of vertices $v_{ab}$ and $v_{(a+1)(b+1)}$.

\begin{definition}
Let $G=(V,E)$ be a connected graph with $n$ vertices and $n$ a multiple of 3. We say that $G$ is \emph{detachable} if there exists a cut $C$ that partitions $G$ in two connected subgraphs of size $n_1$ and $n_2$ such that 3 divides $n_1$ and 3 divides $n_2$. We also say that $C$ detaches $G$.
\end{definition}

\begin{definition}
A region $R$ is \emph{detachable} if its planar dual graph is detachable.
\end{definition}

A very simple fact is that if $R$ is a region of size $n$, then a planar dual graph of $R^\boxplus$ has at least $2n-1$ cycles of size 4. In order to construct a tiling for $R^\boxplus$ we use Lemmas \ref{lem:detach} and \ref{lem:boxplus-cover} given below. But first, we introduce a technical lemma that helps in the proof of Lemma \ref{lem:detach}.

\begin{lemma}\label{lem:cut}
Let $k\geq 3$ be a positive integer and $G$ be any connected graph of size a multiple of $k$ that contains exactly one cycle of length at least least $k + 1$. Then there exists a cut of edges from the cycle that partitions $G$ in two trees whose sizes are both multiples of $k$.
\end{lemma}
\begin{proof}
	
Suppose for the sake of contradiction that any cut of edges in the cycle of $G$ yields two trees: one of size $n_1$ and another of size $n_2$, where $n_1$ or $n_2$ is not a multiple of $k$.

Let $n$ be the number of vertices in $G$ and let $e_0, e_1,\dots,e_k$ be the $k + 1$ edges of the only cycle of $G$, where $n$ is a multiple of $k$. Also fix a vertex $v \in e_0$. It is clear that cutting any two of these edges will partition G in two trees. 

Let $T_{(i,j)}$ be the tree obtained from cutting edges $e_i$ and $e_j$ such that $v$ is a vertex of $T_{(i,j)}$ and let $\overline{T}_{(i,j)}$ denote the other tree, each of size $|T_{(i,j)}|$ and $|\overline{T}_{(i,j)}|$ respectively.  Let $a_1$, $a_2$, ..., $a_k$ be a sequence of $k$ integers where $a_i$ is the size of the tree $T_{(0,i)}$, i.e., $a_i=|T_{(0,i)}|$. Note that by cutting any pair of edges $\{ e_i, e_j \}$, where $i, j>0$, the size of the trees $T_{(i,j)}$ and $\overline{T}_{(i,j)}$ are $|T_{(i,j)}|= n-|a_i - a_j|$ and $|\overline{T}_{(i,j)}| = |a_i - a_j|$, respectively\footnote{Here $|a_i-a_j|$ denotes absolute value.}. In Figure \ref{fig:sizes} we present a pictorial explanation of how to obtain the values of $|T_{(i,j)}|$ and $|\overline{T}_{(i,j)}|$.

\begin{figure}[t]
\centering
\includegraphics[scale=1.5]{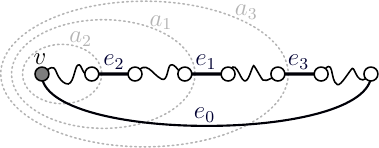}
\caption{Cycle of the graph with sizes $a_1,a_2$ and $a_3$ whose corresponding cuts are indicated by the dashed lines. Thus $T_{(0,2)}$ is the tree inside the dashed line with label $a_2$, $T_{(0,1)}$ is the tree inside the dashed line with label $a_1$ and $T_{(0,3)}$ is the tree inside the dashed line with label $a_3$. From the figure it is easy to see that, for example, if we cut $e_2$ and $e_1$, then $\overline{T}_{(2,1)}=|a_2-a_1|$ and $|T_{(2,1)}|=n-|a_2-a_1|$ where the vertex $v$ is in $T_{(2,1)}$.}
\label{fig:sizes}
\end{figure}


The sequence $a_1,\dots,a_k$ always satisfies one of two cases: (i) the values of $a_i$ modulo $k$ are all different, or (ii) there exists at least one congruent pair $\{a_i, a_j\}$ modulo $k$, \textit{i.e.}, $a_i \equiv a_j \pmod k$ where $i \neq j$. If they are all different, there is some $a_i$ that is congruent to 0 modulo $k$. By cutting the pair of edges $\{e_0, e_i\}$ we have $|T_{(0,j)}|=a_i \equiv 0 \pmod k$ and $|\overline{T}_{(0,j)}|= n-a_i \equiv 0 \pmod k$, which is a contradiction. If we have a congruent pair $\{a_i, a_j\}$, then if we cut the pair of edges $\{ e_i, e_j \}$ we have that $|T_{(i,j)}| = |a_i - a_j| \equiv 0 \pmod k$ and $|\overline{T}_{(i,j)}| = n- |a_i - a_j| \equiv 0 \pmod k$, which is again a contradiction. Therefore, we conclude that there is a cut where $n_1$ and $n_2$ are both multiples of $k$.
\end{proof}

\begin{lemma}\label{lem:detach}
Let $G$ be any connected graph of size $n$ with $n$ a multiple of 3. If $G$ contains a cycle of size at least 4, then $G$ is detachable.
\end{lemma}
\begin{proof}
Suppose that $G$ has at least one cycle. The following procedure finds a cut that detaches $G$.
\begin{enumerate}
\item $C\gets \{\}$.
\item Repeat until exactly one cycle remains in $G$.\label{st:repeat}
	\begin{enumerate}
    \item Pick any cycle $X$ in $G$.
    \item Pick any edge $e$ in $X$.
    \item Delete $e$ from $G$.
    \item Add $e$ to $C$.
    \end{enumerate}
\item From the last cycle in $G$ select two edges $e_1$ and $e_2$ such that $G$ is partitioned in two connected components $G_1$ and $G_2$ each of size a multiple of 3.\label{st:cut}
\item Add $e_1$ and $e_2$ to $C$.
\item Add edges from $C$ to $G_1$ and $G_2$ in such a way that $G_1$ and $G_2$ become induced subgraphs of $G$. All edges that remain in $C$ define a cut that detaches $G$ in $G_1$ and $G_2$.\label{st:x}
\end{enumerate}

The procedure presented above clearly terminates because at each iteration the number of cycles is decreased in one. At the end of Step \ref{st:repeat}, $G$ contains exactly one cycle. We need to show that in Step \ref{st:cut}, there is a cut that partitions $G$ in two connected components each of size a multiple of 3. Indeed, at the end of Step \ref{st:cut} and by letting $k=3$ in Lemma \ref{lem:cut}, there exist two edges $e_1$ and $e_2$ that partitions $G$ in two trees each of size a multiple of 3.

Finally, in Step \ref{st:x}, since $G_1$ and $G_2$ are induced subgraphs of $G$ and there are no edges between $G_1$ and $G_2$, the set $C$ defines a cut of $G$.
\end{proof}

\begin{corollary}\label{cor:tree}
For any planar dual graph $G$ of a region $R$, if $G$ is not detachable, then $G$ is a tree.
\end{corollary}
\begin{proof}
This corollary follows immediately from Lemma \ref{lem:detach} by observing that all cycles in a planar dual graph  $G$ of a region $R$ are of size at least 4.
\end{proof}

\begin{lemma}\label{lem:boxplus-cover}
Let $R$ be a connected region of size $n$ with $n$ a multiple of 3. If $R$ is not detachable, then $R^\boxplus$ has a cover.
\end{lemma}
\begin{proof}
Let $R$ be a connected region of size $n$, with $n$ a multiple of 3, that is not detachable. By Corollary \ref{cor:tree}, the planar dual graph $G$ of $R$ is a tree.

There are four types of vertices in $G$, leaves (degree 1), trunks (degree 2), forks (degree 3), and crosses (degree 4). In order to construct a cover for $R^\boxplus$, we will present a tiling for each type of vertex and then show how to assemble all the parts.

The idea for constructing a tiling is to assign a ``tag'' to each cell of $R$ and its corresponding vertex of $G$ that will help in constructing the entire cover for $R^\boxplus$. In what follows we show how to assign theses tags and later how to use them to construct a cover.

First we consider a leaf, say $\ell$. Since $n$ is a multiple of 3, if we delete $\ell$ from $G$, the size of the resulting graph equals 2 modulo 3. We use the size of the resulting graph modulo 3 as a tag for $\ell$. See Figure \ref{fig:leaf}(a). In the figure we assumed that the cell corresponding to $\ell$ is attached to the region $R$ on its right side. The general rule is that a tag is assigned to the side of the cell that is connected to $R$

For a leaf we  use a tiling pattern of a single tromino as shown in Figure \ref{fig:leaf}(b). The remaining cell in the lower right corner of Figure \ref{fig:leaf} will be covered after assembling the leaf with another vertex. A reflected or rotated leaf uses a reflected or rotated tromino tiling pattern, respectively.

\begin{figure}[htb!]
	\centering
	\subfloat[Leaf tag]{
		\centering
		\hspace{0.1cm}
		\includegraphics[scale=1]{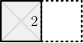}
		\hspace{0.1cm}
	}
	\quad\quad\quad\quad\quad\quad
	\subfloat[Tiling pattern]{
		\hspace{0.8cm}
		\centering
		\includegraphics[scale=1]{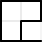}
		\hspace{0.8cm}
	}
\caption{Leaf}
\label{fig:leaf}
\end{figure}

Now we consider a trunk. Here we have two cases, a straight trunk and a bent trunk, see Figures \ref{fig:trunk}(a) and \ref{fig:trunk}(b). We add the tags 1 and 1 to the sides of a trunk because that is the size modulo 3 of both resulting graphs after deleting the trunk. The tiling pattern for both types of trunks is presented in Figures \ref{fig:trunk}(c) and \ref{fig:trunk}(d). In Figure \ref{fig:trunk}(a) note that there are four different ways to place a tromino, however, we use this tiling pattern in a way to enforce that any two trominoes placed on a straight trunk must  have  two  of its cells over the straight trunk. This way, depending on how a straight trunk is connected to the region a correct tiling can be chosen.

\begin{figure}[!htb]
\centering
\subfloat[Straight trunk]{
	\hspace{0.1cm}
	\raisebox{.685cm}{\includegraphics[scale=1]{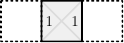}}
	\hspace{0.1cm}
}
\quad\quad\quad\quad
\subfloat[Bent trunk]{
	\hspace{0.1cm}
	\includegraphics[scale=1]{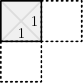}
	\hspace{0.1cm}
}
\quad\quad\quad\quad
\subfloat[Straight trunk tiling]{
	\hspace{0.8cm}
	\raisebox{.685cm}{\includegraphics[scale=1]{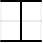}}
	\hspace{0.8cm}
}
\quad\quad\quad\quad
\subfloat[Bent trunk tiling]{
	\hspace{0.8cm}
	\raisebox{.685cm}{\includegraphics[scale=1]{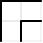}}
	\hspace{0.8cm}
	}
\caption{Trunk.}
\label{fig:trunk}
\end{figure}

Before proceeding to the next type of cell, we show how to use the information of the tags to construct a tiling for the type of cells introduced so far. In Figure \ref{fig:small-example} we present a tiling example of a region with two leafs and one trunk. To construct a tiling we always join a 1 tag with a 2 tag. Note that in Figure \ref{fig:small-example}(c) in order to attach the lower leaf to the trunk, the tiling pattern of Figure \ref{fig:leaf} had to be reflected and then rotated.

\begin{figure}[!htb]
\centering
\subfloat[Region $R$]{
	\includegraphics[scale=1]{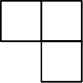}
}
\quad\quad\quad\quad
\subfloat[Tag assignment]{
	\includegraphics[scale=1]{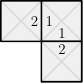}
}
\quad\quad\quad\quad
\subfloat[Tiling of $R^\boxplus$]{
	\includegraphics[scale=1]{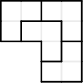}
}
\caption{Tiling example of a region with two leafs and one trunk.}
\label{fig:small-example}
\end{figure}

In a fork there are two cases to consider depending on how a trunk is connected to the region, see Figure \ref{fig:fork}. The tags again are assigned depending on the sizes modulo 3 of the graphs that result after deleting the fork. The tags assignment produces two types of forks called a 2-2-1 fork and a 2-1-2 fork. Note that these tags are the same for the rotated and reflected variants of a fork. Figures \ref{fig:fork}(c,d) present the tiling patterns.

\begin{figure}[!htb]
\centering
\subfloat[2-2-1 Fork]{
	\hspace{0.1cm}
	\includegraphics[scale=1]{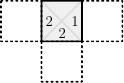}
	\hspace{0.1cm}
}
\quad\quad\quad\quad
\subfloat[2-1-2 Fork]{
	\hspace{0.1cm}
	\includegraphics[scale=1]{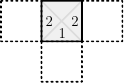}
	\hspace{0.1cm}
}
\quad\quad\quad\quad
\subfloat[2-2-1 Fork tiling]{
	\hspace{0.8cm}
	\raisebox{.685cm}{\includegraphics[scale=1]{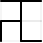}}
	\hspace{0.8cm}
}
\quad\quad\quad\quad
\subfloat[2-1-2 Fork tiling]{
	\hspace{0.8cm}
	\raisebox{.685cm}{\includegraphics[scale=1]{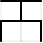}}
	\hspace{0.8cm}
}
\caption{Fork.}
\label{fig:fork}
\end{figure}

Finally, we consider a cross which has two cases as shown in Figure \ref{fig:cross-2}. As in previous cases, the tiling patterns work for the reflected and rotated versions of the trominoes.

\begin{figure}[!htb]
\centering
\subfloat[$2^1$ Cross]{
	\hspace{0.1cm}
	\includegraphics[scale=1]{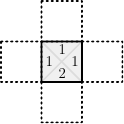}
	\hspace{0.1cm}
}
\quad\quad\quad\quad
\subfloat[$2^4$ Cross]{
	\hspace{0.1cm}
	\includegraphics[scale=1]{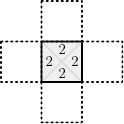}
	\hspace{0.1cm}
}
\quad\quad\quad\quad
\subfloat[$2^1$ Cross tiling]{
	\hspace{0.8cm}
	\raisebox{.685cm}{\includegraphics[scale=1]{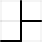}}
	\hspace{0.8cm}
}
\quad\quad\quad\quad
\subfloat[$2^4$ Cross tiling]{
	\hspace{0.8cm}
	\raisebox{.685cm}{\includegraphics[scale=1]{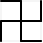}}
	\hspace{0.8cm}
}
\caption{Cross.}
\label{fig:cross-2}
\end{figure}

The general rule for constructing a tiling is thus to use the patterns presented above and joining them by their opposing tags and the use of rotations and reflections of the pattern.

Note that since $R$ is not detachable, its planar dual graph is a tree. This implies that certain cells are connected to other cells only in certain ways. For example, the bent trunk of Figure \ref{fig:trunk}(b) cannot be adjacent to the the fork of Figure \ref{fig:fork} even though they have opposing tags. That is because if they were connected in that way, then $R$ would have a cycle in its dual planar graph. They can, however, be connected in other ways.
\end{proof}

To finish the proof of Theorem \ref{the:detachable}, take any region $R$ and apply the following procedure to decompose $R$ into a set of non-detachable subregions. Find a cut that detaches $R$ in two subregions $R_1$ and $R_2$; if any of these subregions $R_1$ or $R_2$ are detachable, keep finding cuts that detach them in smaller subregions and so on. Repeat this procedure until all subregions are not detachable. By Lemma \ref{lem:detach} it is always possible to decompose $R$ in this way.

Let $R_1,\dots R_k$ be a collection of subregions of $R$ that are not detachable obtained by the procedure of the preceding paragraph. By Lemma \ref{lem:boxplus-cover} each $R_i$ has a cover. Since each $R_i$ was obtained by decomposing $R$ in detachable parts, we can then join the covers of all subregions to obtain a cover for $R$. In Figure \ref{fig:big-region} we present an example of a detachable region with all four types of cells given in the proof of Lemma \ref{lem:boxplus-cover}. Note that if the region $R$ is a tree, then we can find a cut that detaches its dual planar graph using depth-first search.

\begin{figure}[!htb]
\centering
\subfloat[Region $R$]{
	\includegraphics[scale=0.9]{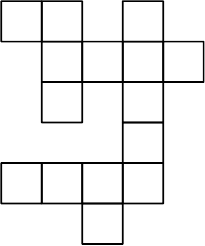}
}
\quad\quad\quad\quad
\subfloat[Tags assignment]{
	\includegraphics[scale=0.9]{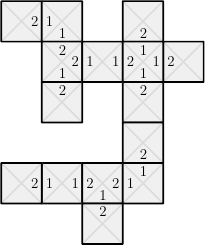}
}
\quad\quad\quad\quad
\subfloat[Tiling of $R^\boxplus$]{
	\includegraphics[scale=0.9]{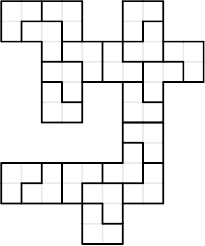}
}
\caption{Example of a detachable region $R$ with a tiling for $R^\boxplus$.}
\label{fig:big-region}
\end{figure}

The proof of Theorem \ref{the:detachable} also gives an efficient algorithm to find covers for any $R^\boxplus$.

\section{Concluding Remarks and Open Problems}\label{sec:conclusions}
In this work we studied the computational hardness of tiling arbitrary regions with L-trominoes. We showed restrictions to the problem that keeps it computationally intractable and identified concrete instances where  an efficient tiling exists. In particular we showed that tiling an Aztec rectangle (and hence, Aztec diamond) with defects is still a hard problem, but in the presence of 0 or 1 defects, a tiling is decidable in polynomial time. Furthermore, even if we restrict the problem of tiling an arbitrary region with $180^\circ$ rotations of  L-trominoes it remains intractable.  We showed, however, that if the region does not contains a so-called ``forbidden polyomino'' as a subregion, then the tiling problem is decidable in polynomial time.

We conclude this paper with some open problems that we consider challenging and that we believe will fuel future research in the subject.
\begin{enumerate}
\item \emph{Hardness of tiling the Aztec rectangle with a given number of defects.} In Section \ref{sec:Aztec} we saw that an Aztec rectangle with 0 or 1 defects can be covered with L-trominoes in polynomial time, whereas in general the problem is NP-complete when the Aztec rectangle has an unknown number of defects; with $2+3k$, for every $k$, an Aztec rectangle cannot be covered because the number of cells is not divisible by 3. It is open if there exists a polynomial time algorithm for deciding a tiling for an Aztec rectangle with a given number of defects.
\item \emph{Tiling of orthogonally-convex regions.} In this work we showed several instances where a tiling can be found in polynomial time. In general, it is open if an orthogonally-convex region with no defects can be covered in polynomial time or if it is NP-complete to decide if a tiling exists.
\item \emph{Enumeration of tromino tilings.} We have not considered the problem of enumerating tromino tilings of the regions described in this paper. In general, there are no such formulas known in the literature for the shapes studied so far. However, we leave it as an open problem for future research to enumerate such tilings. (See \ref{enum} for some discussion.) 
\end{enumerate}

\section*{Acknowledgements}

The authors thank the reviewers and the editor for helpful comments and suggestions.

\appendix

\section{Enumeration of L-tromino tilings of Aztec diamonds}\label{enum}
We have some amount of data for the number of L-tromino tilings of Aztec diamonds (denoted by $T(n)$ for an order $n$ Aztec diamond), which we present in Table \ref{table}. As already commented earlier, it is difficult in general to enumerate polyomino tilings of regions on a finite lattice. However, we can give very simple bounds on the number of such tilings.
\begin{table}
\centering
\begin{tabular}{ll} 
$n$ & $T(n)$   \\ 
\hline
 $1$ &  $0$ \\
 $2$ &  $3$ \\
 $3$ &  $18$ \\
 $4$ &  $0$ \\
 $5$ &  $16856$ \\
$6$ &  $2931525$ \\
 $7$ & $0$ 
\end{tabular}
\caption{Values of number of L-tromino tilings of Aztec diamonds.}
\label{table}
\end{table}

\begin{theorem}
If $n=3k$ for some $k>0$, then we have
\begin{equation}\label{eq:enum}
    T(n)\geq 4T(n-1)+4\sum_{l=1}^{k-1}(l-1)T(3l)+\sum_{l=1}^{k-1}(l-4)T(3l-1).
\end{equation}
\end{theorem}

\begin{proof}
We first notice that any tiling of $\ad(n)$ contains all other tilings of $\ad(m)$ for $0<m<n$ and $m(m+1)\equiv 0 \pmod 3$, as subtilings. This is clear from Figure \ref{fig:enum-1}.
\begin{figure}[htb!]
\centering
\subfloat[Getting a tiling of $\ad(m)$ from $\ad(m-1)$.]{
	\includegraphics[scale=0.5]{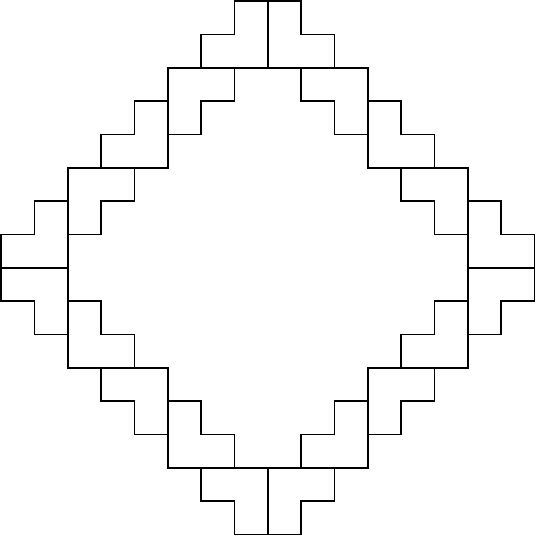}
}
\quad
\subfloat[Getting a tiling of $\ad(m)$ from $\ad(m-2)$.]{
	\includegraphics[scale=0.35]{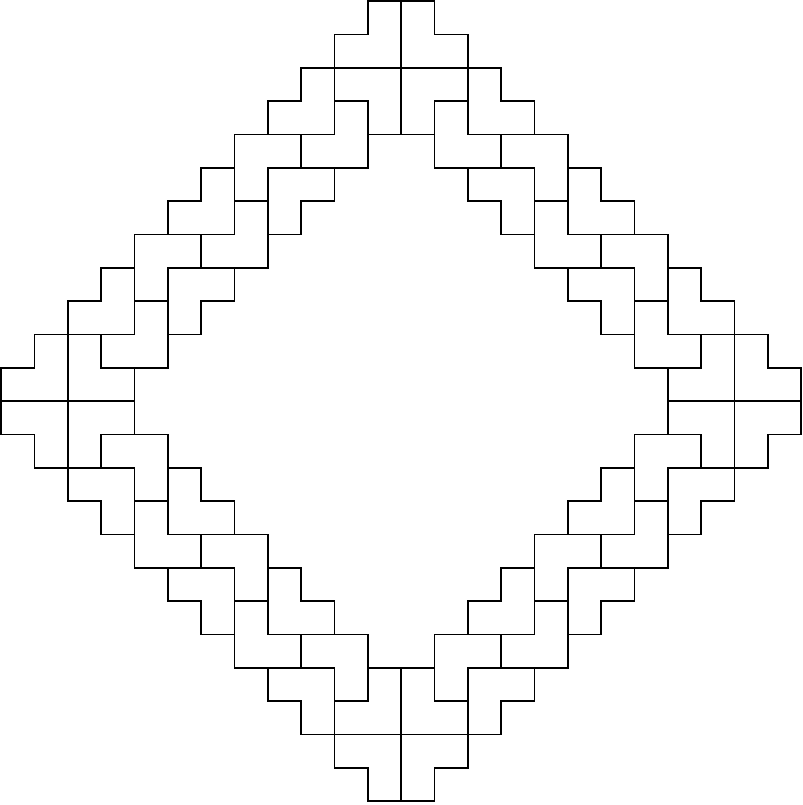}
}
\caption{Tilings using stairs.}
\label{fig:enum-1}
\end{figure}
We further see that $\ad(n)$ can be tiled from any tiling of $\ad(n-1)$ by appending stairs on two adjacent boundary sides. This is shown in Figure \ref{fig:enum-2}. There are $4$ ways to select adjacent sides, which gives the factor of $4$ in the first term of the right hand side of \eqref{eq:enum}.
\begin{figure}[htb!]
    \centering
    \includegraphics[scale=0.5]{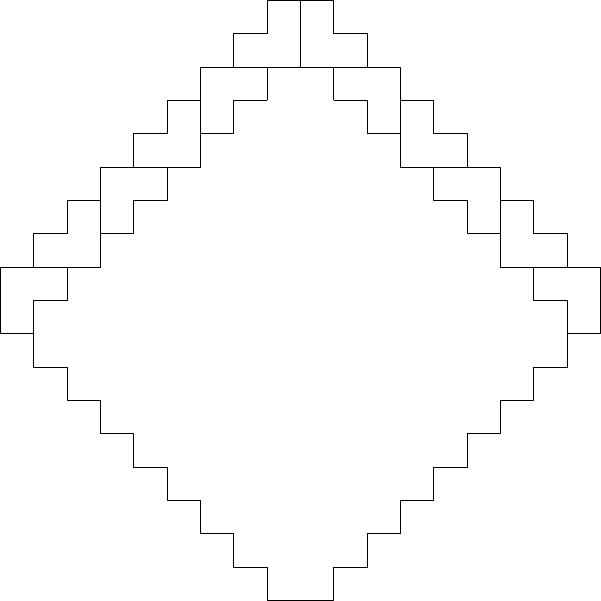}
    \caption{Getting a tiling of $\ad(n)$ from $\ad(n-1)$.}
    \label{fig:enum-2}
\end{figure}

Now let us remove all the subtilings from $\ad(n-1)$. In doing this, we will see that the way we have constructed Figure \ref{fig:enum-2} actually interacts with the construction in Figure \ref{fig:enum-1} to produce $2k$ $2\times 6$ rectangles which can each be tiled in $2$ different ways, for $n=3k$. From here we have the term $4k$ in the bound for different values of $k$. So, we have
\begin{align*}
        T(n)\geq &4\left(T(n-1)-\sum_{l=1}^{k-1}T(3l)-2\sum_{l=1}^{k-1}T(3l-1)\right)\\
    &+4\sum_{l=1}^{k-1}lT(3l)+\sum_{l=1}^{k-1}T(3l-1)(l+T(2)).
\end{align*}
This gives us the bound in \eqref{eq:enum}.
\end{proof}

\begin{remark}
A similar result will also hold for Aztec Rectangles, but with more parameters as well as increased complexity.
\end{remark}

\bibliographystyle{elsarticle-harv}


\end{document}